\documentclass[12pt]{article}
\usepackage[margin=1.0in]{geometry}
\usepackage{amsmath,amsthm, bm,bbm, graphicx}
\usepackage{amssymb, amsfonts}

\usepackage{xcolor}
\usepackage{authblk}

\newcommand{\tr}{\mathrm{Tr}}

\DeclareMathOperator{\PR}{PR}

\usepackage{lineno}

\title{Neural population geometry and optimal coding of tasks with shared latent structure}

\author[1,2]{Albert J. Wakhloo}
\author[2,3]{Will Slatton}
\author[2,3]{SueYeon Chung}

\affil[1]{Center for Theoretical Neuroscience, Columbia University}
\affil[2]{Center for Computational Neuroscience, Flatiron Institute}
\affil[3]{Center for Neural Science, New York University}

\date{}

\begin{document}
\maketitle

\section{Abstract}


Humans and animals can recognize latent structures in their environment and apply this information to efficiently navigate the world. However, it remains unclear what aspects of neural activity contribute to these computational capabilities. Here, we develop an analytical theory linking the geometry of a neural population's activity to the generalization performance of a linear readout on a set of tasks that depend on a common latent structure. We show that four geometric measures of the activity determine performance across tasks. Using this theory, we find that experimentally observed disentangled representations naturally emerge as an optimal solution to the multi-task learning problem. When data is scarce, these optimal neural codes compress less informative latent variables, and when data is abundant, they expand these variables in the state space. We validate our theory using macaque ventral stream recordings. Our results therefore tie population geometry to multi-task learning.

\section{Introduction}

Humans constantly solve different instances of similar problems. We brake at stop signs, stop at red lights, and slow down in crowded streets. We do these things effortlessly and efficiently learn to use new sensory cues to regulate our behavior. This is possible because we are able to recognize overt symbols like road signs, as well as more abstract visual cues like the crowdedness of a street. More generally, humans and animals learn to recognize latent variables in the environment and use them to guide their behavior across contexts and tasks. 

Recent experimental findings have described coding strategies that may underlie this ability. In particular, several studies have described cases where independent variables in the environment are represented along distinct directions of variation in the neuronal activity space \cite{Courellis2023.11.10.566490, nogueira2023geometry, johnston2023semi, bernardi2020geometry, higgins2021unsupervised, chang2017code, flesch2022orthogonal}–e.g., in orthogonal subspaces of the firing rates of a collection of neurons \cite{nogueira2023geometry, flesch2022orthogonal, johnston2023abstract, libby2021rotational, srinath2024orthogonal}. These independent factors have ranged from the contacts of distinct mouse whiskers \cite{nogueira2023geometry}, to more abstract latent variables such as the values of different choices in a decision making task \cite{johnston2023semi}. Neural representations in which distinct environmental variables are represented along independent or orthogonal directions of variation are referred to as factorized or disentangled. Factorized representations were recently shown to emerge in artificial networks trained on multiple tasks \cite{johnston2023abstract} and are thought to support generalization to new contexts as well as efficient learning of new tasks that depend on shared latent variables \cite{bernardi2020geometry, higgins2018towards}.

Similarly, other studies have argued that the brain makes widespread use of cognitive maps to solve these problems \cite{behrens2018cognitive}. These are coding strategies in which environmental variables are represented in the population code in a way that preserves task-relevant relations between them. This idea is supported by a range of findings, for example, studies in which structurally similar tuning profiles emerge in a neural population for distinct types of environmental variables \cite{aronov_mapping_2017, knudsen_hippocampal_2021, nieh_geometry_2021, bao_grid-like_2019, constantinescu_organizing_2016, muhle2023goal}. These variables have ranged from an animal's position \cite{moser2008place}, to the frequency of an auditory stimulus \cite{aronov_mapping_2017}, to more abstract quantities like the amount of evidence accumulated in a decision making task \cite{nieh_geometry_2021}. Here and in the findings referenced above, neural populations represent latent structure in the environment in a way that supports a target behavior. However, defining measures for and understanding why certain neural activity patterns represent latent variables in task-efficient ways remains challenging \cite{urai2022large}.

A promising approach to tying neural activity patterns to such computational goals is to study the geometry of the neural responses \cite{chung2021neural}. Here, the overarching idea is to find which mesoscopic statistics of the population activity contribute to a macroscopic target computation or behavior. In this way, we can gain insight into neural computation without having to give a detailed account of microscopic single unit activity. For example, in the domain of invariant object recognition, recent works analytically tied coding efficiency \cite{chung2018classification, cohen2020separability, wakhloo2023linear, chou2024neural} and few-shot generalization performance \cite{sorscher2022neural} to measurable statistics of the population activity. Recent works have also studied a wide range of motor \cite{ elsayed2016reorganization, russo2018motor, perich2018neural, lanore2021cerebellar, manley2024simultaneous}, sensory \cite{nogueira2023geometry,chang2017code, bao_grid-like_2019, stringer2019high,  she2021neural, froudarakis2020object}, and decision making \cite{johnston2023semi, bernardi2020geometry, knudsen_hippocampal_2021, nieh_geometry_2021} computations and behaviors by analyzing the geometry of neural population responses. 

In this study, we develop an analytical theory for learning binary decision making tasks depending on a common latent structure that directly ties the statistics of neural population responses to generalization performance. While several authors have used linear probes or heuristic geometric measures to analyze population activity in such tasks (e.g. \cite{nogueira2023geometry, bernardi2020geometry}), to the best of our knowledge, there is no cohesive theory that directly ties mesoscopic statistical features of neural activities to task performance in these settings. To fill this gap, we analytically calculate how neural population geometry shapes the generalization error of an agent learning multiple tasks that depend on a common latent structure. To do this, we use a flexible model in which binary classification tasks are formed in an unobserved latent space, and an agent makes decisions by applying a linear readout to a set of neural population activities \cite{johnston2023abstract, gardner1989three, goldt2020modeling, loureiro2021learning, engel2005stat}. 

Using these results, we show how the error is completely determined by a small set of geometric statistics summarizing the relationship between the latent variables and neural responses. We leverage this theory to derive normative predictions describing the optimal geometry that a set of population responses should take. In this way, we show that disentangled representations naturally emerge as an optimal solution to the multi-task learning problem. Furthermore, we show how optimal codes compress less useful information when data is scarce and expand this information in the activity space when data is abundant. We go on to describe how this strategy would be reflected in measurable properties of the eigenspectrum of neuronal responses. Finally, we show that our theory accurately predicts the generalization error of readouts trained on artificial neural network (ANN) data as well as multi-unit recordings taken from the macaque ventral stream \cite{majaj2015simple}, and we demonstrate the use of our geometric measures as a data-analytic tool.

\section{Results}

\subsection{Theory of multi-task learning}

We study the ability of a neural population to support downstream learning of tasks that depend on a common latent structure. Specifically, we consider a setting in which an agent learns binary classification tasks using a set of $ p $ training stimuli. These stimuli are formed from a set of unobserved $d-$dimensional latent variables, $z \in\mathbb{R}^d$ (Fig 1a). The task labels for a given binary classification task are formed by linearly separating the latent space into two pieces using a hyperplane with a normal vector, $ T $. Note that while the separation is linear in the latent space, it may be highly non-linear in the stimulus space. Alongside the latent variables, we consider the activity patterns of $n$ neurons for each stimulus. We denote the vector of activity patterns for each stimuli as $ x\in \mathbb{R}^n$. We quantify how well these neuronal activity patterns support downstream classification by calculating the generalization error of an agent that makes predictions using  a linear readout of these activities, which is learned using a supervised Hebb rule (Methods; Fig. \ref{fig:schem}(e-f)) \cite{engel2005stat}. As described in the Supplemental Material (SM), we calculate the generalization error of this readout both for a single, fixed task as well as the average generalization error across all possible tasks. In this way, we connect the statistical properties of the neural code to the multi-task learning problem. 

While the generalization error for arbitrary distributions of neuronal activities and latent variables is analytically intractable, we show that in many cases, the error only depends on a few key statistics. To do this, we draw from recent work in deep learning theory showing that the generalization error of linear readouts trained on complex distributions can in many cases be well approximated by studying simpler Gaussian models (Methods) \cite{goldt2020modeling, loureiro2021learning, goldt2022gaussian,montanari2023universality}. Thus, we analytically calculate the generalization error using a Gaussian model and show empirically that our theory accurately predicts the generalization error of the linear readout rule when applied to the activations of non-linear multi-layer perceptrons (MLPs; Sec. \ref{sec:MLP}) as well as real neural responses from V4 and IT on visual classification tasks (Sec. \ref{sec:maj}).

\begin{figure}
    \centering
    \includegraphics[width=0.95\textwidth]{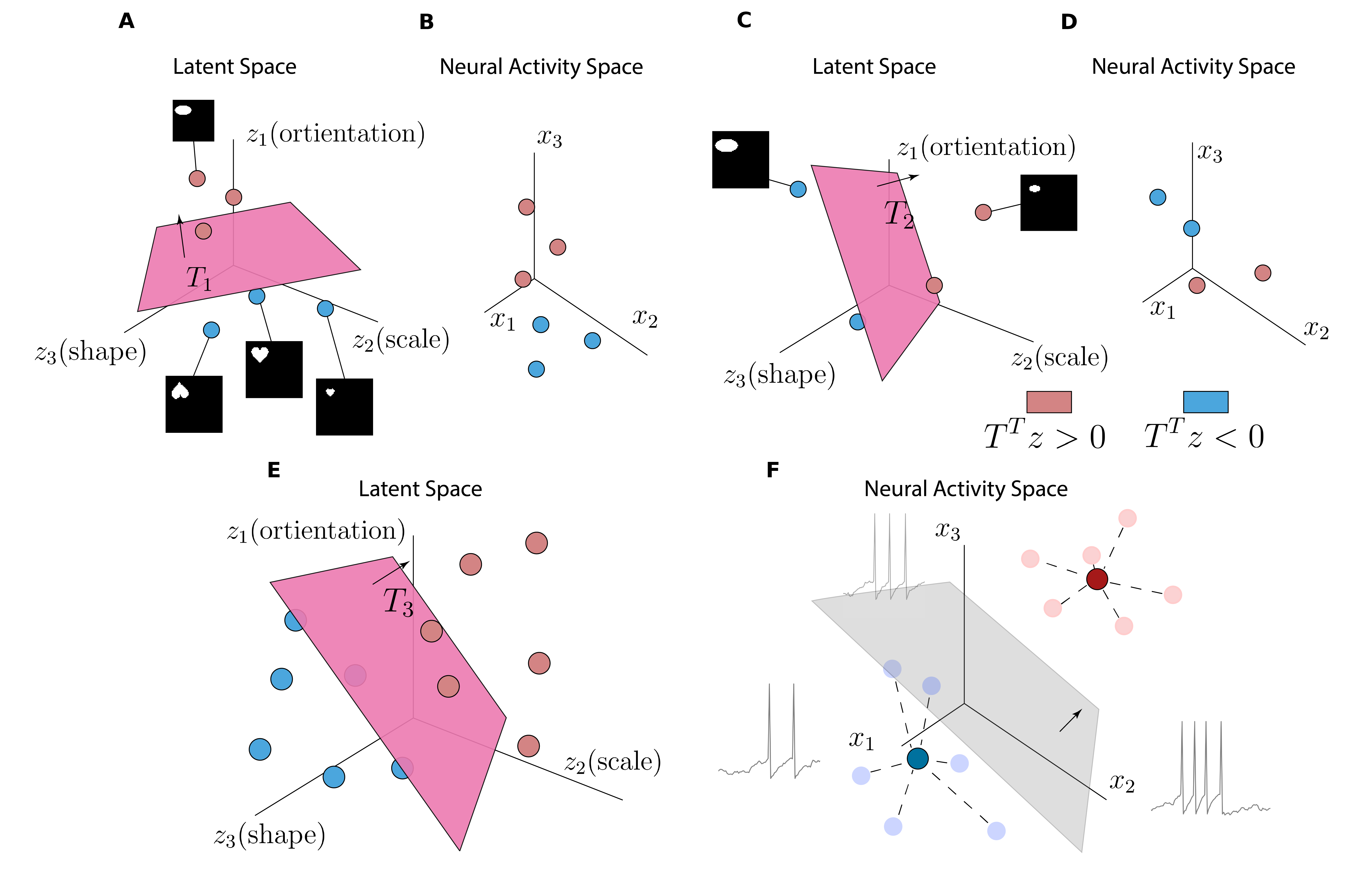}
    \caption{Schematic of the task and model setup using images from the d-sprites dataset as an example \cite{dsprites17}. (a)  Visual stimuli are generated from points in a latent space, and binary discrimination tasks are formed by linearly separating the latent space using a hyperplane with normal $T_1$. (b) Each stimulus elicits a neuronal activity pattern, visualized as points in an activity space. (c) A new binary discrimination task can be formed by separating the latent space using a different hyperplane with normal $T_2$.  (d) As before, stimuli elicit neural activity patterns. (e-f) Schematic of the Hebbian learning rule we consider. Given a set of $p=12$ training stimuli, subsequent decisions are made using a linear readout of the neural activity patterns, in this case the firing rates of 3 neurons. When the number of positive (red) and negative (blue) labels is balanced, this linear readout corresponds to using a hyperplane whose normal points in the direction of the difference in the means of positive and negative examples .}
    \label{fig:schem}
\end{figure}

\begin{figure}
    \centering
    \includegraphics[width=0.95\textwidth]{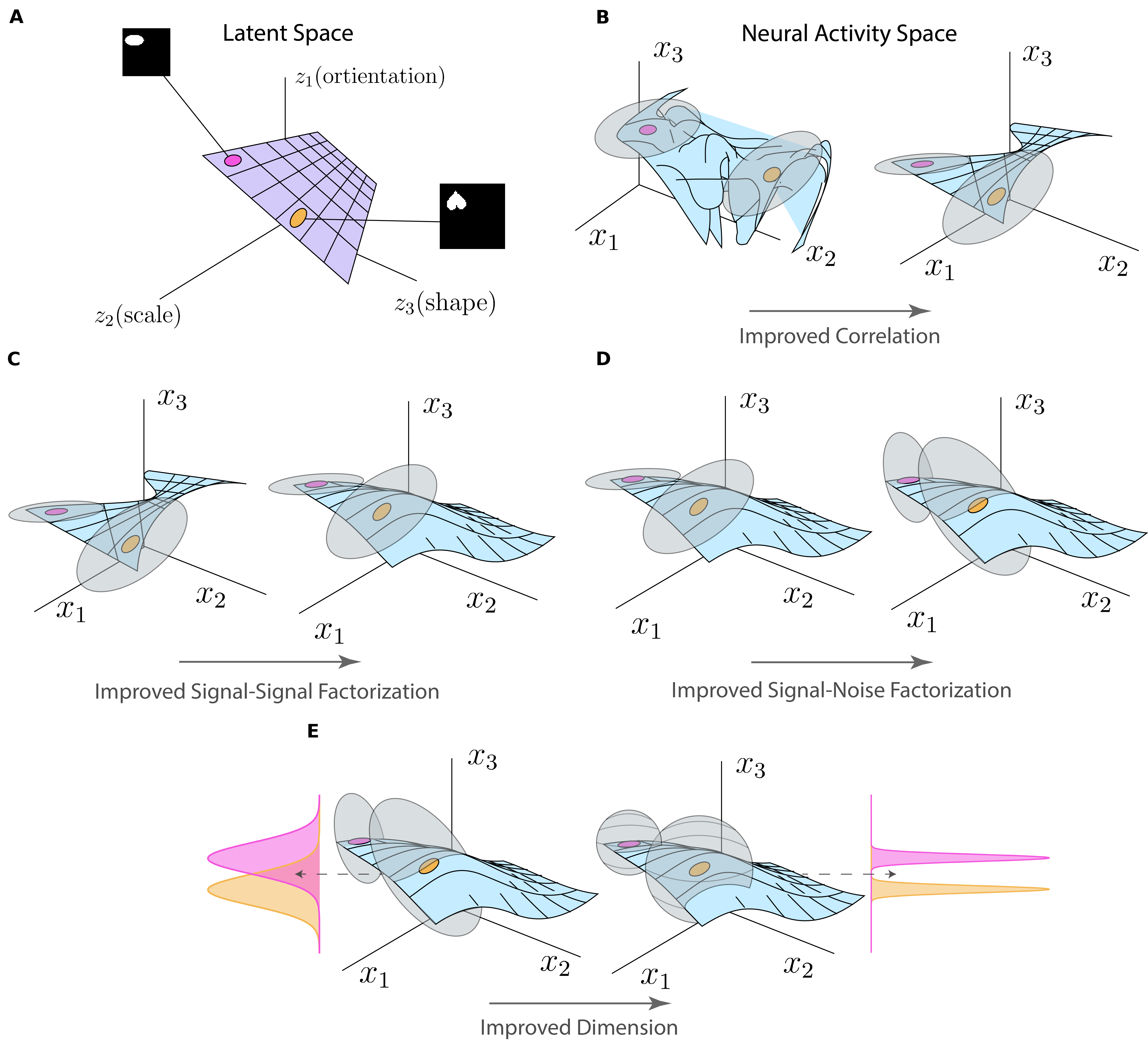}
    \caption{Schematic of the geometric terms. We visualize different possible neuronal activity patterns elicited by the same set of stimuli. (a) A small slice of the latent space from which stimuli are generated. (b) Visualization of neural activity patterns with low (left) and high (right) total correlation. When the correlation is high, the relative distances between points in the latent space are approximately preserved in the neural state space. (c) Signal-signal factorization (SSF). When the SSF is low, different latent variables are represented along overlapping directions, and when it is high, independent directions in the latent space are represented along approximately orthogonal directions in the neural state space. (d) Signal-noise factorization (SNF). When the SNF is low, the noise distribution (grey ellipses) around a point in the firing rate space falls along the coding directions. When it is high, the noise distribution is orthogonal to these directions. (e) Neural dimension. In higher dimensional representations, the neural activity and associated noise distribution occupies more directions in the state space, shown here as 2d (left) vs. 3d (right) noise distributions. As the dimension increases, the projection of a sample of neural activity onto a given direction becomes increasingly concentrated, supporting generalization performance \cite{sorscher2022neural}.}
    \label{fig:geometry}
\end{figure}

Using this simplified model, we derive a formula for the average generalization error. As described in the Supplemental Material, we prove that under certain conditions, the generalization error can be written as a strictly decreasing function of four geometric terms, which we schematize in Fig. \ref{fig:geometry}. We now discuss each of these one by one: 
\begin{itemize}
  \item \textbf{Neural-latent correlation:} $c.$ This term is a normalized sum of all covariances between neurons and latent variables. Geometrically, the neural-latent correlation measures how sensitive the population responses are to variations in the latent space.  
  \item \textbf{Signal-signal factorization (SSF):} $f$. This term measures the alignment between the coding directions of distinct latent variables. The signal-signal factorization term favors neural coding schemes that represent distinct latent variables along orthogonal directions in the state space. Moreover, this term encourages these representations to devote equal variance to distinct coding directions–i.e., to form a whitened representation of the latents. 
  \item \textbf{Signal-noise factorization (SNF):} $s$. This term measures the magnitude of the noise that lies along the coding directions of the latent variables. Ideally, any noise present in the neural responses should lie in directions which are orthogonal to the directions representing the latent variables.   
  \item \textbf{Neural dimension:} $\PR.$ The participation ratio of the neural responses measures the effective number of dimensions that the population activity spans. When all else is equal, higher dimensional responses are preferred, as neuronal noise is less correlated from trial to trial \cite{sorscher2022neural}. 
\end{itemize}
Finally, we note that the contribution to the error of the neural dimension and neural-latent correlation decays to 0 with the number of training samples (Methods). On the other hand, the error stemming from the signal-signal and signal-noise factorization represents an \emph{irreducible error} that does not depend on the number of training samples (SM). Intuitively, in the few-shot regime, the main concern is maximizing the amount of total signal in the neural code, while minimizing the impact of noise. These two aspects are primarily controlled by the total correlation and dimension, respectively. On the other hand, in the many-shot regime the main concern becomes keeping the representations of distinct latent variables separate from one another and separate from noise directions. These features of the code are respectively controlled by the two factorization terms. 


\begin{figure}
    \centering
    \includegraphics[width=0.95\textwidth]{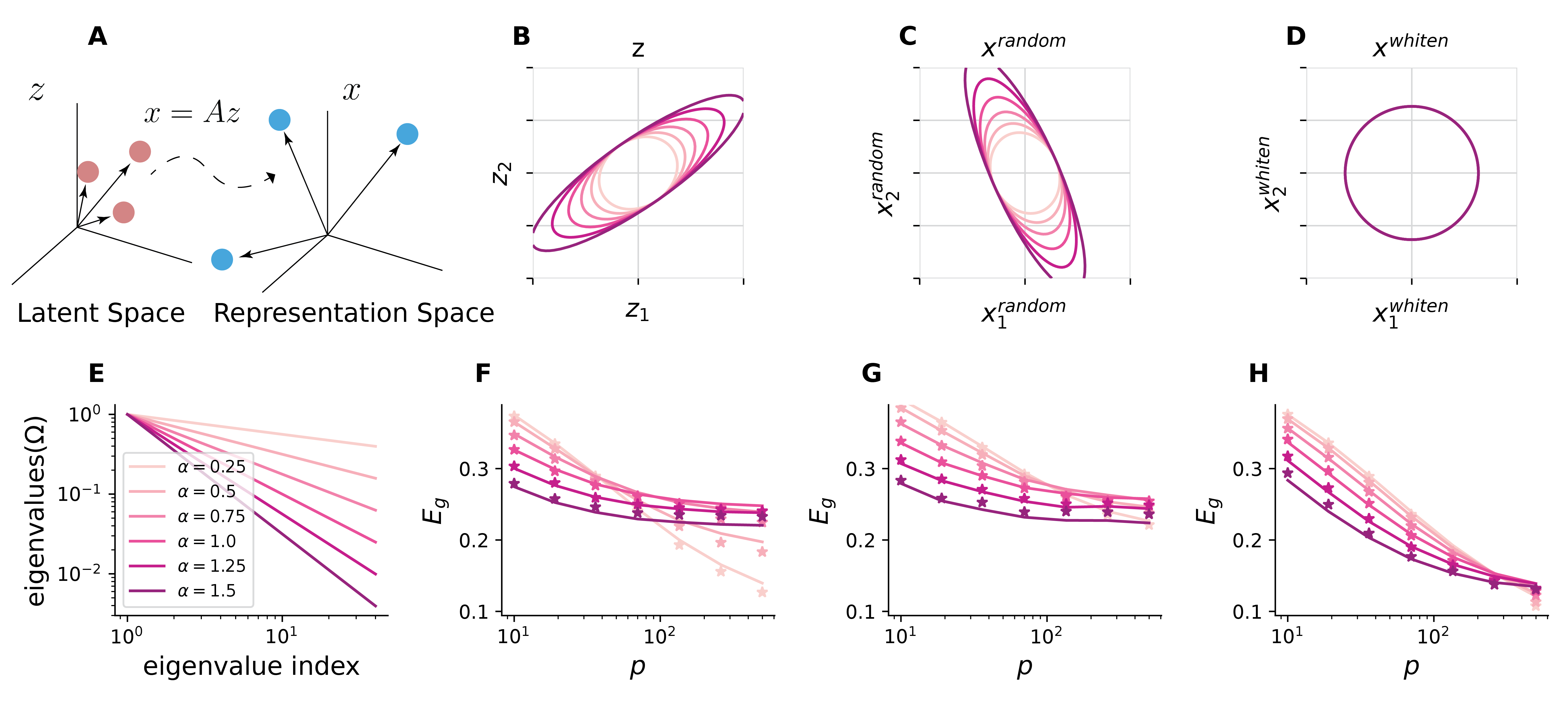}
    \caption{Theory predicts empirical generalization error in Gaussian model with power law covariance spectra. (a) Schematic illustrating simulation setup. Gaussian latent variables $z$ are used to generate task labels, and predictions are formed using a linear transformation of the latent variables, $x=Az$. (b-d) Two typical units for the (b) latent variables, (c) random high dimensional projection, and (d) whitened transform for various values of the spectral decay exponent, $\alpha.$ (e) Eigenvalues of the latent covariance, $\Omega$, for different decay rates, $\alpha.$ (f-h) Multi-task generalization error as a function of training samples, $p$, for (f) the latent variables themselves, (g) the random projection, and (h) the whitened transform.}
    \label{fig:cloud}
\end{figure}

We begin by validating our theory numerically on data drawn from a Gaussian model. Here, we find that our theory yields an excellent agreement with numerical simulations for a wide range of  values (Methods). In these simulations, we sample the latents $ z_\mu $ and neural responses $ x_\mu $ from a multivariate Gaussian. We set the covariance matrix of the latents to have a spectrum that decays as a power law. Since the proof of our main theorem assumes that the spectra of the covariances decays slowly, this allows us to parametrically study violations of our main theorem's assumptions (Fig. \ref{fig:cloud}a-e; SM). After sampling the latents, we form the neural responses either by taking a random high dimensional projection of the latent variables, or by applying a whitening transform to the latent variables (Fig. \ref{fig:cloud}b-d; see Methods). We then calculate the empirical generalization error across a set of binary classification tasks which are formed by shattering the latent space as above. 

As shown in Fig. \ref{fig:cloud}f-h, our theory yields an excellent fit to numerical simulations. Importantly, the theoretical prediction holds all the way down to the few shot learning regime in which the number of training samples $ p $ is small, despite the theory being derived in the limit of large $p$. Furthermore, we find that the theory predicts the empirical generalization error well for a relatively small number of neurons $n $ and latent dimensionality $ d $, as well as latent variables whose covariance has an eigenspectrum that decays relatively quickly \cite{stringer2019high} (Methods). These results suggest that our theory can be applied to a wide range of datasets and is informative of the few shot learning regime.

\subsection{Optimal representation of latent variables}

\begin{figure}
    \centering
    \includegraphics[width=0.95\textwidth]{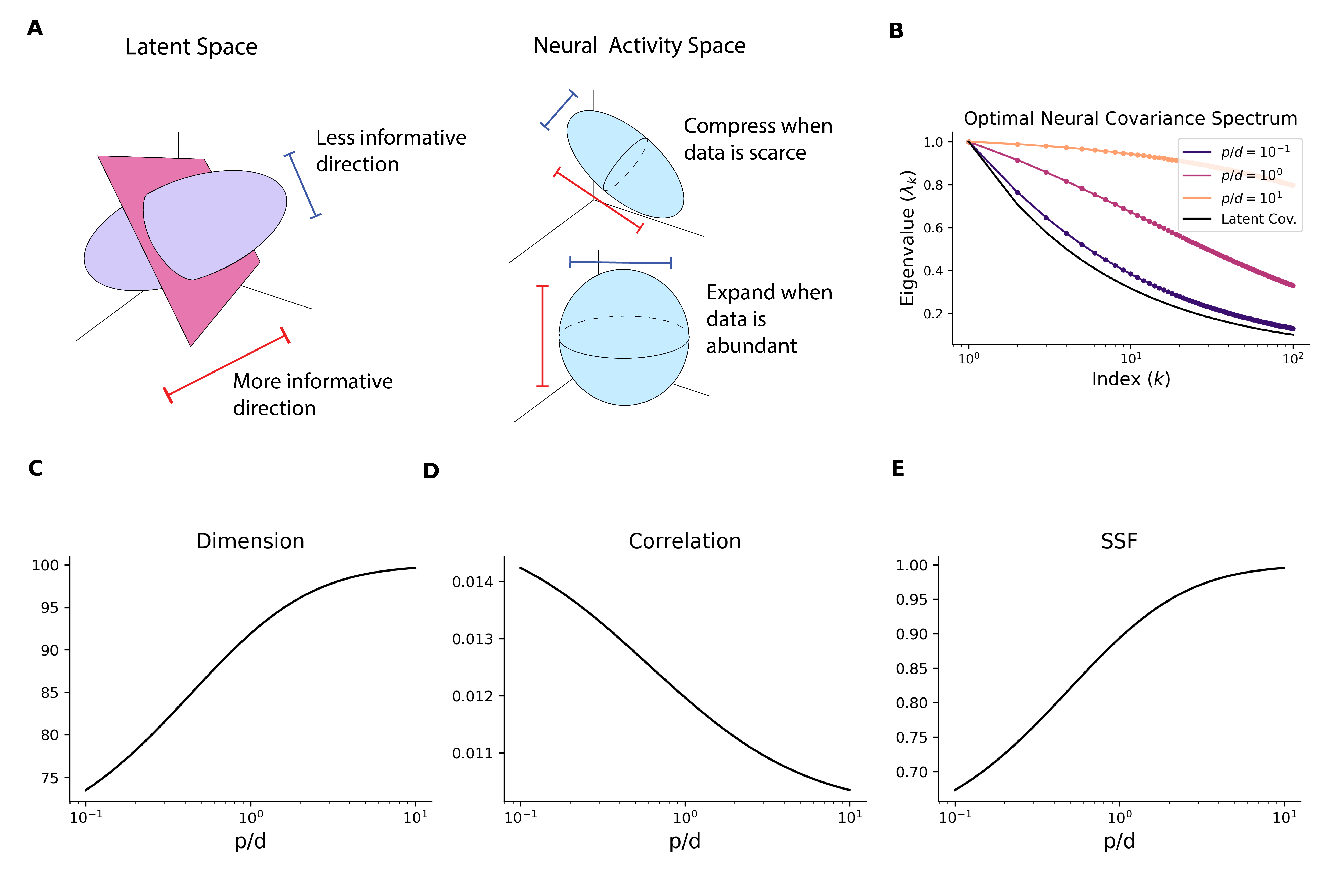}
    \caption{Optimal representational geometry as a function of training samples and latent structure. (a) In our task setup, directions in the latent space that have little variance are, on average, less informative of the task labels (Methods). Optimal neuronal representations of these latent variables map independent directions in the latent space to independent directions in the neuronal space. The amount of variance corresponding to less informative directions in the latent space is small when data is scarce and is large when data is abundant. (b) Eigenvalues of the optimal neural covariance as a function of the number of samples vs. latent dimension, $p/d $. We show the eigenvalues of the latent variables' covariance in black. Markers correspond to results obtained by optimizing our formula for the generalization error numerically, and solid lines correspond to our formula for the optimal code's spectrum (Methods). As the number of samples increases, the spectrum of the optimal neural code becomes increasingly flat, indicating the expansion strategy described above. (c-e) Geometric terms of the optimal neural representation for various values of $p$, given the same latent covariance. (Note that we do not plot the signal-noise factorization, as it diverges for the optimal representation for all $p$.) Surprisingly, the total correlation \emph{decreases} with $p$.  }
    \label{fig:opt-rep}
\end{figure}

Our theory provides an ideal framework to pose normative questions regarding multi-task learning. In our model, latent variables which have less variance contain, on average, less information about the task labels (Fig. \ref{fig:opt-rep}a). This allows us to probe how task structure and the number of available training samples determine which neural representations are optimal in the sense of achieving the lowest possible generalization error. In this section, we study the geometry of optimal neural representations, given some latent variables and a fixed training dataset size. 

We show analytically that the optimal representation disentangles latent variables into orthogonal subspaces. More precisely, we show that each orthogonal direction of variation in the latent space maps onto an orthogonal direction of variation in the neural state space. Thus, disentangled representations emerge as the optimal solution to the multi-task learning problem. 

In addition to being disentangled, we find that the optimal code compresses less informative variables when data is scarce and expands these variables when data is abundant. Importantly, while each latent is represented in an orthogonal direction, the variance along each direction is not equal. As shown in Fig. \ref{fig:opt-rep}a, when training data is limited, less informative latent variables correspond to directions in the neuronal activity space that have small variance. With increasing training data (i.e., as $ p $ increases), the amount of variance dedicated to these less informative latent variables grows. This reflects a strategy in which the optimal code compresses less informative variables at small training set sizes $ p $, while at large sample sizes, these variables are expanded in the state space (Methods). Speaking informally, the optimal code only starts paying  attention to the less informative latent variables when there are enough samples to learn their relevance to a given task.

We trace these features of the optimal neural code back to the eigenvalues of the neuronal covariance matrix and our geometric terms. As shown in Fig \ref{fig:opt-rep}b, the eigenspectrum of the optimal code becomes increasingly flat as $p$ increases. This reflects the fact that more and more variance is being dedicated to the less informative directions in the latent space (i.e., directions with smaller variance). Turning to the geometric terms, we can see that this trade off appears as a trend in which the participation ratio and signal-signal factorization term of the optimal code both increase with $ p $, while the correlation \emph{decreases} (Fig. \ref{fig:opt-rep}c-e). That is, optimal neural codes become increasingly high dimensional as the number of available training samples increases. A key normative prediction from our theory is therefore that the task-related variability of neural responses becomes increasingly high dimensional as an agent learns to perform tasks that depend on a complex latent structure.

\subsection{Geometry of multi-task learning in MLPs}
\label{sec:MLP}

\begin{figure}
    \centering
    \includegraphics[width=0.95\textwidth]{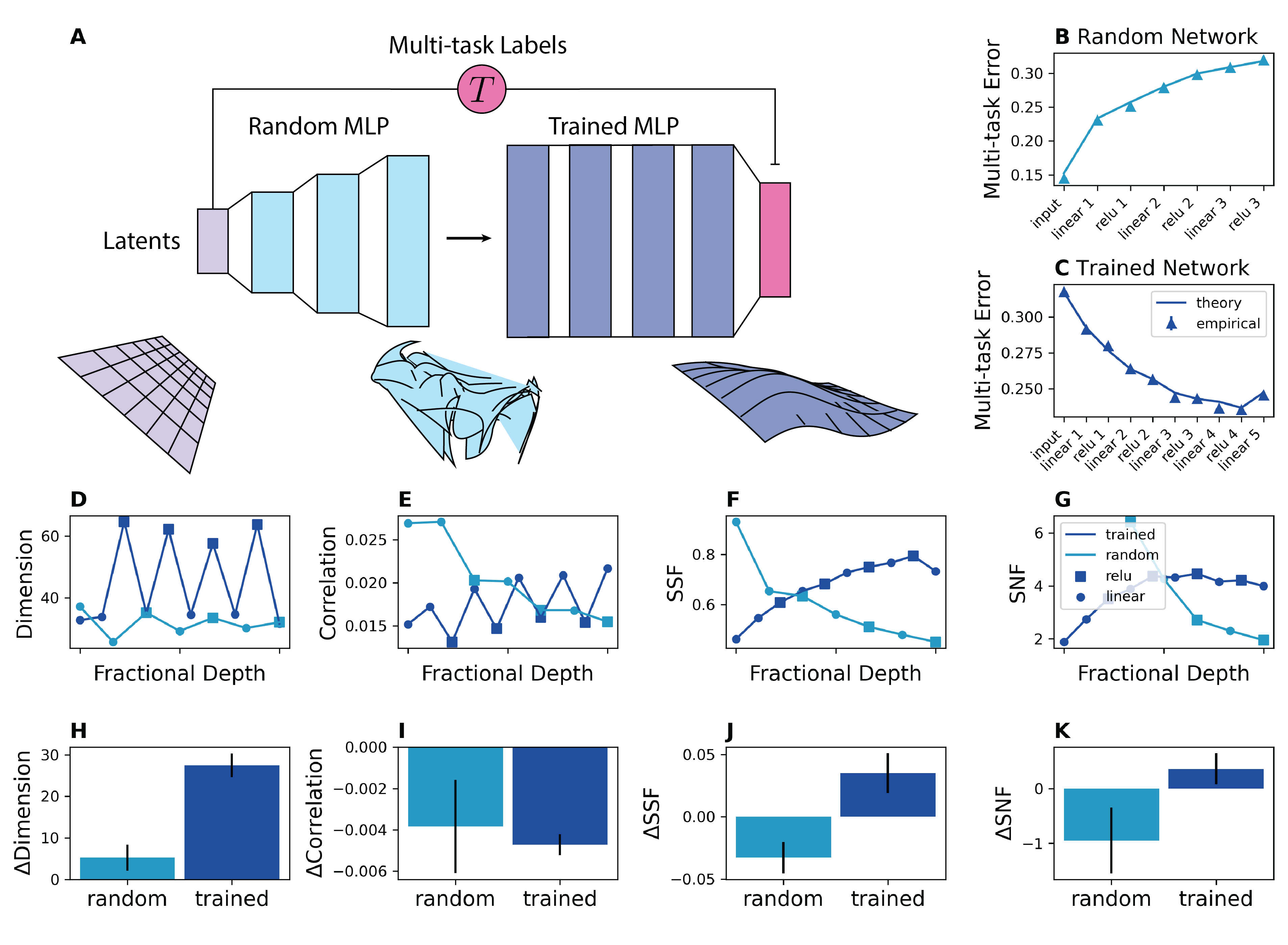}
    \caption{Theory predicts generalization error of the Hebb rule in trained and random MLPs. (a) Schematic of the experiment. Latent variables $z$ are randomly shattered to generate task labels. These latents are passed through a random MLP (light blue) and are then used as inputs to train a 3 hidden layer MLP (dark blue) on the multi-task binary classification problem using stochastic gradient descent. (b-c) After training, we sample a new set of latents and teacher vectors and calculate the generalization error of the Hebb rule on each layer of the (b) random and (c) trained network. Theoretical predictions closely track empirical errors, and the trained network achieves a lower error in intermediate layers. (d-g) Geometric terms across layers for the random (light blue line) and trained (dark blue line) networks. Linear layers are marked by circles and relu layers by squares. Interestingly, the error only slightly changes across linear and relu layers of the same model stage, in spite of sharp changes in the geometry. (h-k) Average change in the geometry before and after the application of relu. Here we show the mean and standard deviation of the difference between each geometric term before and after the relu nonlinearity is applied. In the trained network, the application of relu consistently causes increases in the dimension and signal-signal factorization, as well as decreases in the correlation.}

    \label{fig:mlp1}
\end{figure}

In practice, neural data is complex and may be non-Gaussian. To test our theory in this regime, we now apply our theory to both random and trained non-linear multi-layer perceptrons (MLPs).  As shown in Fig. \ref{fig:mlp1}a, we first sample a set of zero-mean Gaussian latent variables $ z $. We then feed these latent variables through a random MLP with increasing hidden layer sizes. After generating this set of stimuli, we sample $500 $ random $ T $ vectors, and use these to generate a set of task labels (Methods). These labels and data points from the random MLP are then used to train a downstream 4-hidden-layer MLP, which predicts the label for each task using a linear readout of its penultimate layer. This is reminiscent of the hidden manifold modeling framework \cite{goldt2020modeling}. Finally, we sample a new set of latent variables along with a new set of tasks and calculate the generalization error of the Hebb rule when applied to the representations at different layers of the trained and untrained MLPs (Methods). This setup allows us to validate our theory on non-linear transformations of Gaussian latent variables. 

As shown in Fig. \ref{fig:mlp1}b-c, we find good agreement between our theoretical predictions and the empirical generalization error. Examining the layer-by-layer generalization errors we find that the random MLP successfully tangles neural representations of the latents, from the perspective of the downstream Hebbian readout, as the generalization error increases through the random MLP. On the other hand, we find that the trained MLP successfully manages to learn the latent structure of the task, as demonstrated by the fact that the generalization error of the Hebb rule drops sharply through the layers of the network. Overall, these results provide compelling evidence that our analytical formula holds well beyond settings in which the neural responses are Gaussian variables. 

Turning to our geometric decomposition of the error, we find several interesting trends across layers and non-linearities (Fig. \ref{fig:mlp1}d-g). In the random MLP, we see that the correlation decreases each time the non-linearity is applied, while the signal-noise and signal-signal factorization terms gradually decrease through the network. We also find that the dimension of the responses stays more or less constant through the network. 

In the trained network, we find that linear and relu layers orchestrate a trade off between the geometric terms that together lead to an overall decrease in multi-task generalization error. As shown in Fig \ref{fig:mlp1}h-k, the neural dimension spikes each time the nonlinearity is applied, at the cost of the total correlation. Conversely, at each linear layer, the correlation sharply rises while the dimension falls. On the other hand, we find that both the signal-signal and signal-noise factorization terms monotonically increase through the penultimate layer of the network. Importantly, this pattern is not present in the random network. Thus, the trained MLP learns to use the nonlinearity to increase the dimension of the representation as well as the factorization of the latent variables, while simultaneously squashing particularly harmful noisy directions of variability in the input. Intuitively, the linear layers learn to orient signal-unrelated features into parts of the state space that are zeroed out by the non-linearity \cite{keup2022origami}. While these are sharp changes in the geometry, the overall generalization error across layers exhibits only minor fluctuations between most linear and relu layers, highlighting the limitations of a generalization error-based approach to analyzing network activity.

We now study the dynamics of these geometric terms through training. To do this, we track the layer-wise generalization error and geometry through a single epoch of training (Methods). As shown in Fig. \ref{fig:mlp-dyn}, we find that the generalization error in this context decays through training across layers. In the linear layers, we find that improved generalization is driven by a steady rise in the signal-signal and signal-noise factorization as well as small increases in the total correlation and dimension. On the other hand, in relu layers, the correlation \emph{decreases} through training, with improved generalization being driven by increases in the dimension as well as the signal-signal and signal-noise factorizations across both early and late layers. We find similar trends for MLPs with a $\tanh$ non-linearity (SM Figs. 2\&3). These geometric trends mirror the effects of adding more samples on the optimal representation: As more samples become available, the factorization terms and dimension grow while the correlation shrinks. In other words, the dynamics of MLP representations in non-linear layers qualitatively mimic the behavior of the optimal neural code described by our theory.

\begin{figure}
    \centering
    \includegraphics[width=0.95\textwidth]{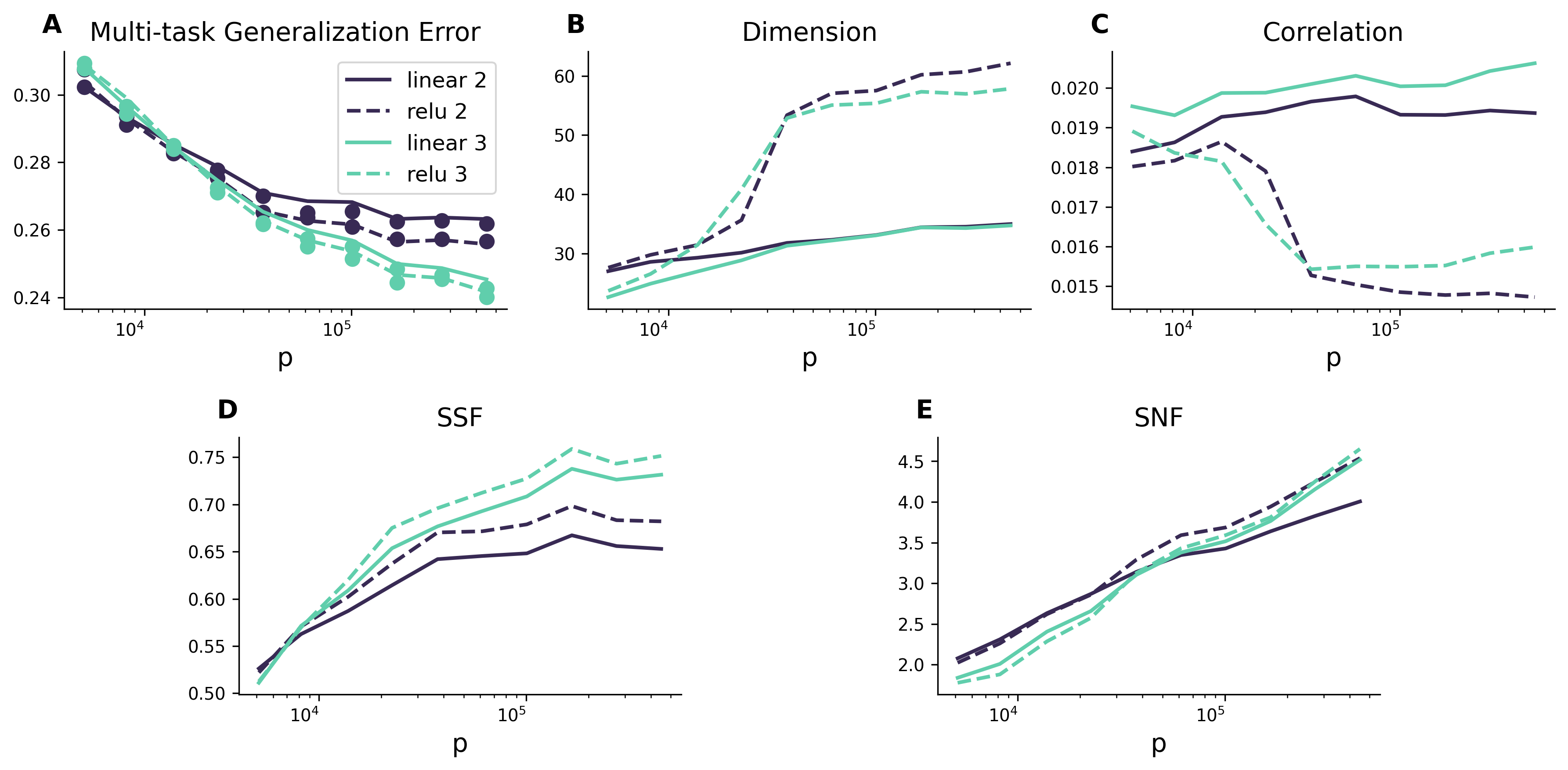}
    \caption{Evolution of generalization error and representational geometry through training. (a) Generalization error of the Hebb rule applied to four different layers on a previously unseen set of tasks. We examine layers from the early stages of the MLP (black) as well as late stages (green) for both linear (solid line) and relu (dashed line) layers. We can see that the error decreases across all layers with the number of training samples, $p$. (b) Dimension term. While the dimension of the neural activations increases for all layers through training, the differences are most pronounced in the relu layers. (c) Correlation term. This term decreases through training for relu layers and slightly increases for linear layers. (d) Signal-signal factorization term. (e) Signal-noise factorization term. Both of these terms uniformly increase through training across the network, though the former begins to plateau at the end of training. These findings mirror the effect of increasing the number of training samples on the optimal neural code.}
    \label{fig:mlp-dyn}
\end{figure}

\subsection{Predicting readout performance of macaque visual representations}
\label{sec:maj}
Having considered non-linear artificial neural networks, we now apply our theory to biological neural data. To do this, we draw from pre-existing multi-unit recordings from macaque V4 and IT taken while 2 monkeys viewed visual stimuli (Methods; Fig. \ref{fig:maj}(a-c)) \cite{majaj2015simple}. The stimuli used in these experiments included images of 64 objects taken from 8 categories and were generated by modifying $d = 6$ continuous latent variables, that included the size, position, and angle of the object.  This allows us to form binary classification task labels by shattering this 6-dimensional latent space on subsets of the data corresponding to individual object categories–e.g., a given task could involve separating images of chairs on the left side of the screen from those on the right. We test our theory by calculating the generalization error of the Hebbian readout rule applied to the V4 and IT neural responses.

We find that our formula accurately predicts the generalization error of the monkey neural responses across both recording sites. Furthermore, as shown in Fig. \ref{fig:maj}d, both V4 and IT achieve better generalization error than the raw pixels, and IT achieves a better generalization error than V4. \cite{hong2016explicit}. Applying our geometric decomposition of the generalization error to these data, we find that the dimension is lower in V4 than in either IT or the pixel data \cite{sorscher2022neural} (Fig. \ref{fig:maj}e). Furthermore, we find that the correlation steadily increases from the raw pixels to V4 and IT, indicating an increased overall signal. We also find that the level of signal-signal factorization improves from V4 to IT, in line with previous results \cite{lindsey2023factorized}. Interestingly, we find that the signal-signal factorization is higher in the pixel space than in either brain region. Finally, we find that the signal-noise factorization is far lower in the pixels and V4 than in IT. This suggests that in IT, latent-unrelated variability overlaps less with the coding directions than in V4 or the pixels (Fig. \ref{fig:maj}e-h).  Taken together, these results provide strong support for our theory and demonstrate the applicability of our metrics as a tool for tying the geometry of neuronal population responses to the computational objective of multi-task learning. 

\begin{figure}
    \centering
    \includegraphics[width=0.95\textwidth]{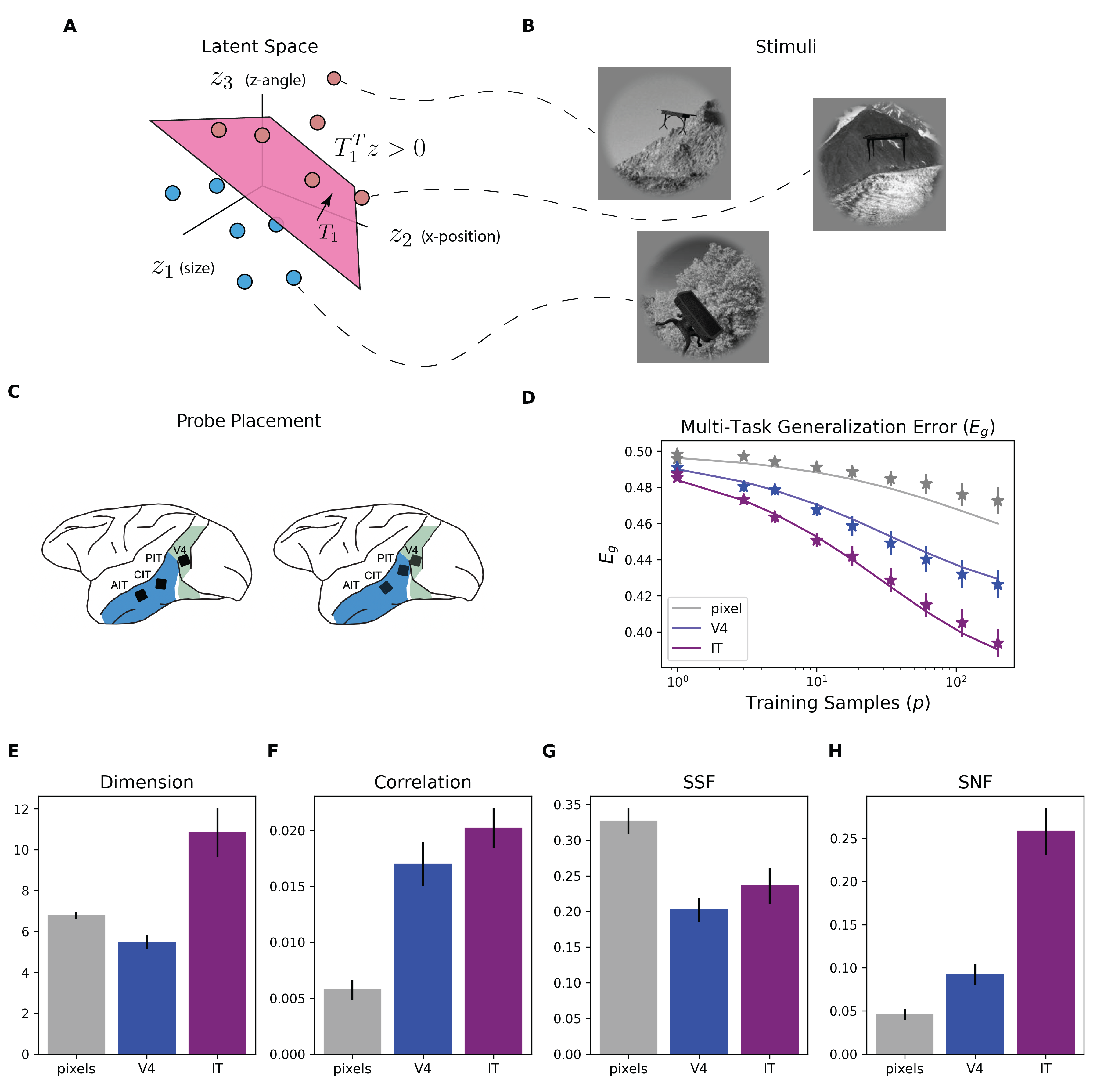}
    \caption{Theory predicts multi-task error in macaque V4 and IT data. (a-b) Example stimuli and tasks. Visual stimuli included 64 objects grouped into 8 categories and were generated by modifying $d=6$ continuous latent variables that included object size, position, and angle. We form binary classification tasks on subsets of the data coming from the same category–e.g. all images from the "Tables" category. (c) Probe placement. Figure adapted with permission from \cite{majaj2015simple}. (d) Generalization error across pixels and brain regions calculated empirically (markers) and using our formula for the generalization error (solid line). (e-h) Geometric terms across pixels and neural responses. We can see that V4 representations are lower dimensional than either the pixels or IT, that the correlation improves throughout, and that the signal-signal factorization is highest in the pixel space, though it improves from V4 to IT. Note too that the signal-noise factorization sharply rises from V4 to IT, suggesting that signal-unrelated variability overlaps less with the coding directions in the latter region.} 
    \label{fig:maj}
\end{figure}

\section{Discussion}

In this work, we analyzed a model of learning multiple tasks that share a common latent structure, connecting neural population geometry to the multi-task learning problem. In this model, binary classification tasks are generated by shattering a latent space uniformly. By varying the spectrum of the latent covariance matrix, we were able to model situations in which certain latent variables are, on average, more useful for downstream tasks than others. We then calculated the generalization error of a linear readout applied to a set of neural activations responding to stimuli that were generated from these latent variables. In this way, we evaluated how well a neural representation is able to extract information about latent variables from the stimuli and use this for arbitrary downstream tasks. 

To connect the mesoscopic statistics of neural activation patterns to generalization performance, we decomposed our formula for the generalization error into four geometric terms. These terms summarize the relationship between population activity and the latent variables and completely determine generalization performance in our task. First is the dimension of the entire population activity, as measured by the participation ratio \cite{jazayeri2021interpreting}. Second, we introduced the total correlation between neuronal activity and the latent variables, representing the overall sensitivity of neuronal responses to changes in the latent space.  Third, we introduced the signal-signal factorization, measuring the degree to which each latent variable is coded in an orthogonal direction in the neural state space, with adequate variance devoted to each coding direction. Finally, we defined the signal-noise factorization, which summarizes the amount of signal-unrelated variability that lies along the directions in state space used to represent the latent variables. Hence, this term summarizes the degree to which noise correlations affect the readout considered here \cite{averbeck2006neural,bartolo2020information, cohen2009attention, ni2018learning}. Together, these four terms completely characterize the generalization performance of the linear readout applied to the neural population responses. 

We next analytically calculated which geometries minimize the generalization error, given a fixed latent structure and number of available training samples. This falls in line with studies on efficient coding in which predictions of sensory physiology are derived from statistics of the natural environment \cite{ganguli2014efficient, atick1990towards,doi2012efficient}. In this way, we found that disentangled representations naturally emerge as an optimal solution to the multi-task learning problem. Specifically, we found that the eigenvectors of the neuronal and latent covariances match the singular vectors of the cross-covariance (Methods, SM). This is reminiscent of work done on the successor representation and grid cells \cite{stachenfeld2017hippocampus, fang2023neural, sorscher2023unified, momennejad2017successor, momennejad2020learning}, and it would be interesting to see if there are additional constraints under which grid cells or other coding properties emerge as an optimal code in our model \cite{xie2023task}.  Future work can also extend our theory to consider sets of tasks that compete with one another as in previous studies of multi-task learning \cite{musslick2018efficiency, petri2021topological}. Moreover, given recent findings regarding multi-task learning and dynamical motif reuse in recurrent networks \cite{yang2019task, driscoll2022flexible}, it would be interesting to extend our theory to a dynamical setting in which predictions are formed using an appropriate readout of time-varying firing rates.

We then described how the optimal neuronal code shifts from being low to high dimensional as the number of available samples increases. This finding reflects a strategy of compressing less useful information when data is scarce and expanding it when data is abundant. It would be very interesting to test this normative prediction in behaving animals. This could be done, for example, using multi-unit recordings taken while an animal learns to perform decision making tasks in which multiple latent variables (e.g. distinct context cues) must be used to make a decision. Our theory predicts that task-related neural variability in the early stages of training should primarily track highly informative latent variables. As training progresses, we predict that task-relevant neural variability will increase in dimension as it begins to track the less informative latent variables in the task. 

It is interesting to compare these findings to previous work reporting that higher dimensional representations of individual object classes are preferred for invariant object recognition in the few shot learning setting \cite{sorscher2022neural}. Although the contribution to the error of the neural dimension scales as 1/ p in our formula just as in this work, the additional contribution to the error of the signal-signal and signal-noise factorization terms together with the constraint of  of the covariance matrix remaining positive semi-definite, leads to optimal representations with \emph{lower} dimension in the few shot regime. These results highlight the fact that different geometric terms may compete with each other in ways that cannot be directly read off from generalization error equations when there are additional constraints imposed on the system. Another related line of work has considered the role of neural dimension in random pattern separation \cite{babadi2014sparseness, litwin2017optimal, muscinelli2023optimal}. In these settings, the optimal dimension can decrease when neural noise increases, and future work can examine whether a similar relation holds in the multi-task learning setting considered here. Finally, we note that theories of least squares estimation suggest that similar phenomena occur for regression tasks with more complex readouts \cite{bordelon2020spectrum, canatar2021spectral, bordelon2022population, canatar2023spectral}. We suspect that a similar spectral expansion trend of optimal coding schemes holds for other common readout mechanisms in both classification and regression tasks, and it would be interesting to investigate these questions in future work. 

To test our theory, we next considered a setting in which a set of Gaussian latents are fed through a random network, and task labels are generated as above in the latent space. A four hidden layer MLP was then trained on the outputs of the random network on these tasks \cite{goldt2020modeling}. We then evaluated each layer of the trained network on previously unseen tasks using the scheme described above. Here, we found that intermediate layers of the trained network successfully learned to represent the latent variables in a way that could be used for downstream tasks, in line with previous work \cite{johnston2023abstract}. Furthermore, we found a strong correspondence between our theoretical predictions and the empirical generalization errors. Applying our geometric decomposition of the error, we found that linear and relu layers obtained similar generalization errors but with sharply different underlying geometries. Specifically, relu layers yielded a large boost to the dimension at the cost of the total correlation, while the signal-signal and signal-noise factorizations steadily improved through the network. This suggests a strategy in which linear layers attempt to place neuronal noise in regions of the state space that are then squashed by the non-linearity \cite{keup2022origami}. Tracking the evolution of these terms through training, we found that the correlation in relu layers \emph{decreases} through training, while the dimension, signal-signal factorization, and signal-noise  factorization increase. This mirrors the behavior of an optimal code as more samples are presented as well as the behavior of deep linear networks trained on certain regression tasks  \cite{saxe2013exact,saxe2019mathematical}. It would be interesting to apply our geometric decomposition to analyze disentanglement in more complex artificial networks. 

Finally, we applied our theory to electrophysiological recordings from macaque V4 and IT \cite{majaj2015simple}. Specifically, we used our theory to predict the generalization error of a linear readout of neural responses trained on a series of binary classification tasks. We formed task labels using latent variables describing object position, size, and orientation. As before, our theory showed excellent agreement with empirical generalization errors of linear readouts applied to the neural responses. We found that both regions showed better readout performance than the pixels themselves \cite{hong2016explicit}. While we found a low signal-signal factorization in these neural responses relative to the pixels, we did find that signal-signal factorization improved from V4 to IT \cite{lindsey2023factorized}. Moreover, we found a sharp increase in the signal-noise factorization from the pixels and V4 to IT, suggesting that latent-unrelated variability may become increasingly orthogonal to the coding directions through the ventral stream. It would be very interesting to repeat such an analysis using data from visual regions specialized for localization. Overall, these results demonstrate the applicability of our theory as a data-analytic tool. 

While we focused on tasks that came from shattering a continuous latent space, future work could extend this theory to consider distributions of latents and tasks that more closely mirror common  experimental settings. Specifically, one could repeat our general calculation, while restricting the distribution of task vectors to a handful of relevant directions or restricting the distribution of latent variables to be fixed to discrete values. This would be particularly interesting in settings where only particular groups of tasks or latent variable values are relevant for an animal to consider. Such a calculation would allow one to tie the population geometry to the linear readout performance on families of tasks and latent distributions which are tailored to specific experimental settings. We hope to pursue this line of research in subsequent work. 


\section{Methods}

\subsection{Model of multi-task learning}
\label{sec:setup}
We model a setting in which an agent learns a set of binary classification tasks by performing a linear readout on a set of neural population activities. Formally, we assume that each stimulus is associated with a $d$-dimensional latent variable $ z_\mu \in \mathbb{R}^d $, with $ 1 \leq \mu \leq p $ denoting the sample index. The labels for a specific task are formed by shattering the latent space using a hyperplane with a normal vector, $ T $. Thus, we have that the binary classification task labels, $y_\mu,$ satisfy the relation, $y_\mu=\mathrm{sign}(T \cdot z_\mu) $, so that each $ T $ vector defines a specific classification task. Associated with each latent, we also consider $n$-dimensional neural activity patterns, $x_\mu \in \mathbb{R}^n$. From these data, the agent then forms predictions of new data points using a supervised Hebbian readout rule \cite{engel2005stat}. More precisely,  given a new stimulus associated with latent variables $z_+$ and firing rates $x_+$, the agent forms a prediction $\hat y_+$ using the rule

\begin{gather}
    \hat y_+ = \mathrm{sign}(w \cdot x_+), 
\quad 
    w= \frac 1 p \sum_\mu y_\mu x_\mu . 
\end{gather} 
When the labels are balanced, this corresponds to using the difference in the mean activities for positively and negatively labeled examples (Fig. \ref{fig:schem}f). We evaluate how well the neural code can support downstream classification by calculating the generalization error of these predictions, averaged across different tasks (i.e., different $T$ vectors), though we also calculate the error for a fixed $T$ along the way (SM). 

While we validate our theory on a wide range of data, we obtain our analytical results using a simplified Gaussian model. This is inspired by work on the Gaussian Equivalence Principle (GEP) in deep learning theory \cite{goldt2022gaussian, goldt2020modeling, loureiro2021learning}. Formally, we assume that each pair of neural responses and latent variables, $(x,z)$ are jointly zero-mean Gaussians with covariance matrices:

\begin{gather} 
	\mathbb E [xx^\top]= \Psi, \quad \mathbb E [xz^\top] = \Phi, \quad \mathbb E[zz^\top] = \Omega.
\end{gather}
We can see that $ \Psi $ describes the neuron-neuron covariances, $ \Phi $ contains the covariances between single-unit responses and variations in the latent variables, and $ \Omega $ describes covariances between latent variables in the dataset. Note that since the $T$ vectors are chosen randomly from a Gaussian distribution, the latent covariance determines which directions in the latent space are, on average, most informative of the task labels. Directions in the latent space which have a significant amount of variance are typically more informative of the task labels in this setup. This model corresponds to a variant of the popular student-teacher model \cite{gardner1989three, loureiro2021learning}. Insofar as the GEP holds, the neuronal and latent (cross-)covariances fully specify the generalization error of the linear readout. 

\subsection{Geometric decomposition of generalization error across tasks}

We prove a formula for the generalization error of the linear readout, averaged across different binary classification tasks (i.e., different $T$ vectors). To do so, we consider the limiting case in which the number of neurons, latent dimensions, and training samples are all large and of comparable size. Furthermore, we assume that the covariance matrices given above satisfy certain spectral properties (SM). Under these assumptions, we show that the average generalization error, $E_g$, is a decreasing function of the four geometric terms. Formally, we have
\begin{gather}
    E_g = \frac 1 \pi \tan^{-1} \bigg(\sqrt{ \frac{\pi}{2pc^2 \mathrm{PR}(\Psi)} + \frac {1} {f} + \frac 1 s -1 }\bigg),
    \label{eq:eg}
\end{gather}
where we have introduced the total neural-latent correlation $ c $, the signal-signal factorization $f$, the signal-noise factorization $s$, and the dimension of the population responses as measured by the participation ratio, $ \mathrm{PR}(\Psi)$. These terms correspond to statistics of the covariance matrices specified above: 

\begin{align}
    c&=\frac{\tr(\Phi\Phi^\top)}{\tr(\Psi)\tr(\Omega)}
    \\
    \PR(\Psi)&=\frac{\tr(\Psi)^2}{\tr(\Psi^2)}
    \\
    f &= \frac{\tr(\Phi\Phi^\top)^2}{\tr(\Omega)\tr(\Phi^\top \Phi \Omega^{-1} \Phi^\top \Phi)}
    \\
    s &= \frac{\tr(\Phi\Phi^\top)^2}{\tr(\Omega)\tr(\Phi^\top(\Psi - \Phi\Omega^{-1}\Phi^\top)\Phi)}
\end{align}
Since the function $ F(w) = \frac 1 \pi \tan^{-1}\big(\sqrt{w-1}\big)$ is strictly increasing, we can see that each of these geometric terms should be made as large as possible. As discussed in the SM, the term $f$ measures the overall degree of orthogonality between coding directions, while $s$ measures the amount of noise that lies along the signal directions. The neural noise is described by the noise covariance matrix, $\mathrm{cov}(x|z) =\Psi - \Phi\Omega^{-1}\Phi^\top,$ which appears in the definition above. Importantly, this neural noise may include variability that is related to latent variables that are not measured experimentally. Motivated by this observation, we describe how these two factorization terms can be collapsed into a single factorization term in the SM. Note that in writing Eq. \eqref{eq:eg}, we have separated a term that depends on the number of training samples, $1/[c^2\PR(\Psi)]$ from a term which is \emph{independent of the sample size}, $1/f + 1/s$. Thus we can see that the correlation and dimension terms become less important as $p$ grows.

\subsection{Gaussian simulations}

To  test our theory on data points which violate the assumptions of our theory, we began by sampling from a finite Gaussian model. Specifically, we drew latent variables, $z_\mu$, from a multivariate Gaussian distribution whose covariance matrix had eigenvalues that decayed as a power law with rate $\alpha$. Specifically, we set: $ \omega_i = 5i^{-\alpha} $, where $ \alpha $ was the power law of the spectrum, and $\omega_i$ is the $i$th eigenvalue of the latent covariance.   The neural responses were then given by the formula $ x_\mu = A z_\mu $ for a random, $ n\times d $ Gaussian matrix $ A $ with i.i.d. elements, or by applying a whitening transform $ x_\mu = \Omega^{-1/2} z_\mu $. We set $ n = 80 $ and $ d= 40 $ for these simulations. For a fixed training set, we sampled $ N_{task} = 300 $ task $ T $ vectors and calculated the generalization error across all tasks using a set of new latent variables.  Finally, we averaged over this entire procedure 30 times to generate the markers in Fig \ref{fig:cloud}. 

\subsection{Optimal codes}

We derive which neuronal codes achieve the lowest generalization error, given a fixed number of samples to train on and a fixed latent structure. Specifically, we calculate which neuron-neuron and neuron-latent covariance matrices,  $\Psi$, $\Phi$ achieve the lowest multi-task generalization error, given fixed a fixed latent covariance $\Omega$ and training set size $p$ (SM). We do this by optimizing Eq \eqref{eq:eg} with respect to $ \Psi $ and $ \Phi $ subject to the constraint that the entire covariance matrix between neurons and latents be positive semi-definite, a necessary condition for the code to be realizable. Using this approach, we find that the left and right singular vectors of $ \Phi $ are the eigenvectors of $ \Psi $ and $ \Omega $, respectively for the optimal code. This shows that independent directions in the latent space map directly onto independent directions in the firing rate space (SM). Furthermore, we obtain the following simple formula for the eigenvalues of the optimal neural code. Denoting the eigenvalues of $ \Psi $ and $ \Omega $ as $ \psi_i $ and $ \omega_i $ respectively, we have up to a permutation symmetry: 
\begin{gather}
    \psi_i =C  \frac{\omega_i}{2p\omega_i + \pi\sum_k \omega_k},
\end{gather}
where $C$ is an arbitrary constant. As $ p $ grows, we can see that the spectrum becomes flatter, reflecting the expansion strategy, while as $ p $ shrinks, the spectrum decays faster and faster, reflecting the fact that less informative directions are compressed in the state space (Fig. \ref{fig:opt-rep}c).

To validate our calculation, we numerically calculated the optimal code by optimizing on the space of positive semi-definite matrices. Since the full covariance matrix is positive semi-definite, there must exist matrices $ X_1 $ and $ X_2 $ such that
\begin{gather}
  L = \begin{pmatrix} \Omega^{1/2} & 0 \\ X_1 & X_2\end{pmatrix},
  \quad 
  LL^T = \begin{pmatrix} \Omega & \Phi^T \\ \Phi & \Psi \end{pmatrix}
\end{gather}
The space of possible $ X $ matrices is unconstrained, so we simply optimize Eq \eqref{eq:eg} with respect to the $X_i$ matrices and calculate $ \Psi $ and $ \Phi $ after the fact.

\subsection{MLP Experiments}

We used random and trained MLPs to test several predictions from our theory using explicitly non-Gaussian artificial neural response data. To generate these data, we first sampled a set of $ d=40 $ dimensional latent variables from a multivariate Gaussian distribution with eigenvalues $ \omega_k = k^{-0.2} $. For these experiments, we sampled $ 5 \cdot 10^5$ latent variables independently from this distribution as a training dataset. Using these latent variables we generated a set of $ N_t = 500$ tasks by randomly shattering the latent space. Latent variables were then passed through a $ 3$ layer perceptron with randomly initialized weights. We chose the size of each intermediate layer to be twice the size of the previous one to ensure the dimension of the representation did not decrease. Each layer was composed of a linear transform, followed by a relu non-linearity (see SM Figs. 2\&3 for results using a tanh non-linearity). After sampling the task labels and passing the latents through the random MLP, we trained a $ 4 $ hidden-layer MLP to predict task labels from the random MLP responses. The trained network was similarly composed of linear-batchnorm-relu blocks, and we used the adam optimizer to train the network \cite{kingma2017adam} through a single epoch. This setup corresponds to a multi-task version of the hidden manifold modeling framework studied in deep learning theory \cite{goldt2020modeling}.

We evaluated the generalization error through layers and training using a new set of  latents and tasks. Specifically, we sampled $10^3$ latent variables from the same distribution as well as a new set of $300$ randomly selected tasks. To calculate the theoretical generalization error , we then used the network representations of these new latent variables to calculate the geometric metrics and evaluate the theoretical generalization errors  using Eq. \eqref{eq:eg} with $p=300$.   To evaluate the empirical generalization error, we randomly split the new set of latents into a set of 300 training points and 700 test points. This process was repeated to calculate the average test error of the Hebb rule across each layer and time point in training. Finally, we averaged over train/test splits to generate the markers shown in Fig. \ref{fig:mlp1}. 

\subsection{Macaque analyses}

We drew from a publicly available dataset containing multi-unit recordings from V4 and IT taken from 2 monkeys \cite{majaj2015simple}. These recordings were taken as the monkeys viewed visual stimuli as described in \cite{majaj2015simple} and contained 88 V4 and 168 IT neural sites, though we reproduce all results projecting IT and pixel responses down to 88 randomly chosen dimensions (SM Fig. 4). The stimuli for this task were drawn from a generative model in which a total of 64 objects coming from 8 categories were displayed against varying backgrounds. These images were generated by varying $ d =6 $ continuous latent variables that controlled the object size, angle, and position. Latent variables were drawn iid from a mean zero uniform distribution.

We tested the linear decodability of these latent variables from the neural firing rates using the scheme described in Sec. \ref{sec:setup}. For the raw pixels, we first carried out a Gaussian random projection onto a $5,000$ dimensional space before applying this scheme. Specifically, for each task condition and for each of the 8 object category types, we formed task labels for $N_t=300 $ tasks by randomly shattering the latent space. For each condition and category type, there were $ 320 $ examples per condition and object category type. We z-scored the latent variables to ensure that there was no especially informative latent variable. We then formed predictions using the supervised Hebbian readout described in Sec. \ref{sec:setup} after mean centering the neuronal firing rates. This procedure was repeated over 15 different train/test data splits. To generate the results in Fig. \ref{fig:maj}, we then averaged the generalization error across the $N_t$ classification tasks, data splits, conditions, and image categories. Error bars denote the standard error of the mean for the estimated generalization error across the $ c = 8 $ category types. In SM, we also show the geometry and generalization errors for each of the categories individually (SM Figs. 5\&6), as well as the error obtained by pooling all categories together (SM Fig. 7).

\section*{Acknowledgements} 

We thank Thomas Yerxa, Chi-Ning Chou, Arnav Raha, and Jenelle Feather for helpful discussions and comments on an earlier version of this manuscript. This work was funded by the Center for Computational Neuroscience at the Flatiron Institute of the Simons Foundation. S.C. is also partially supported by a Sloan Research Fellowship, and a Klingenstein-Simons Award. All experiments were performed on the high-performance computing cluster at the Flatiron Institute.

\bibliographystyle{unsrt}
\bibliography{references}

\end{document}


\maketitle

\tableofcontents

\section{Proof of generalization error formula}

Here we state and prove our formula for the generalization error.

\begin{theorem}
    Let $x_{n,\mu} \in \mathbb{R}^n,z_{n,\mu} \in \mathbb{R}^d$ with $1\leq\mu\leq p$ be a sequence of random variables drawn independently from a Gaussian distribution with mean zero and covariance matrices:

    \begin{gather}
    C_n= 
        \begin{pmatrix}
        \Psi_n & \Phi_n
        \\
        \Phi_n^\top & \Omega_n
        \end{pmatrix},
    \end{gather}
    where $C_n$ is a positive-definite matrix with sub-matrices $\Psi_n \in \mathbb{R}^{n\times n}, \Omega_n \in \mathbb{R}^{d\times d}, \Phi_n \in \mathbb{R}^{n\times d}$, and the ratios $d/n = \alpha $ and $p/n = \beta$ are held fixed.  In addition, let $x_{n,+}, z_{n,+}$ be a pair of  samples drawn independently from the same distribution, and let, $r_n$ be a sequence of $d$-dimensional vectors drawn uniformly from the surface of the sphere of radius $\sqrt{d}$. Assume that there exist positive constants $c_1, c_2$ such that

    \begin{gather}
        \frac{c_1}{n} \leq \psi_{n,i}, \omega_{n,i}, \phi_{n,i} \leq \frac{c_2}{n},
    \end{gather}
    where we use lower-case Greek letters to denote the eigenvalues or singular values of the respective matrices. Define the Hebbian readout $w := \frac 1 p \sum_\mu \sgn(\langle z_{n,\mu}, r_n \rangle) \langle x_{n,\mu}, x_{n,+} \rangle,$ and let $\Theta(\cdot)$ denote the Heaviside step function.  Under these assumptions the generalization error is given by

    \begin{gather}
    \mathbb E_{x_n,z_n,r_n}\big\{ \Theta(-\langle r_n, z_{n,+}\rangle \langle w, x_{n,+} \rangle)\big\}  
    \\
    =  
    \frac 1 \pi \tan^{-1}\bigg( \sqrt{
    \frac{\tr(\Omega_n) \big[
    \frac \pi p \tr(\Psi_n^2)\tr(\Omega_n) 
    + 2 \tr(\Phi_n^\top \Psi_n \Phi_n) 
    ]}{2\tr(\Phi_n^\top \Phi_n)^2}
    -1}
    \bigg) 
    + O(n^{-1/2}) .
    \nonumber
    \end{gather}
\end{theorem}

\begin{proof}

We start by defining the random variables

\begin{gather}
    \gamma^\mu_n = \sgn(\langle z_{n,\mu}, r_n \rangle) \langle x_{n,+}, x_{n,\mu} \rangle, 
\end{gather}
so that the generalization error may be written 

\begin{gather}
    \mathbb E_{r_n} \mathbb E_{x_{n,+}, z_{n,+}} \mathbb E_{x_{n,\mu}, z_{n,\mu}} \Theta\bigg(-\sgn(\langle z_{n,+} r_n \rangle) \sum_\mu \gamma_n^\mu \bigg).
    \label{eq:gamm}
\end{gather}

Note that we have used Fubini's theorem to separate the expectations over the training and test points, as well as the teacher vector $r_n$. 

\indt Our goal is now to show that the distribution of the random variable $\sum_\mu \gamma^\mu_n$ is approximately Gaussian. We do this using the Berry-Esseen theorem, which allows us to carry out the inner expectation in Eq. \ref{eq:gamm}.  We can obtain the relevant moments needed to apply this theorem using the identity  

\begin{gather}
    \mathbb E_{p,q \sim \mathcal N (0, \Sigma) } 
    \sgn(q)p = \sqrt{\frac{2}{\pi \sigma}} \kappa , 
    \quad \Sigma:= \begin{pmatrix}
        \zeta & \kappa \\ \kappa & \sigma 
    \end{pmatrix}, 
\end{gather}
together with 

\begin{gather}
    \mathbb E_{h\sim \mathcal{N}(0, v)} |h|^3 = \frac{2\sqrt 2}{\sqrt{\pi}} v^{3/2}.
\end{gather}
The mean and variance are simply 

\begin{gather}
    \mathbb E_{x_{n,\mu}, z_{n,\mu}} \gamma^\mu_n = \sqrt{\frac{2}{\pi r_n^\top \Omega_n r_n} }x_{n, +}^\top \Phi_n r_n, 
    \\
    \mathbb E_{x_{n,\mu}} (\gamma^\mu_n -\mathbb E_{x_{n,\mu}, z_{n,\mu}} \gamma^\mu_n )^2= x_{n,+}^\top \Psi_n x_{n,+}
    - \frac{2 (x_{n,+}^\top \Phi_n r_n)^2}{\pi r_n^\top \Omega_n r_n}.     
\end{gather} 

It suffices to bound the third central moment as

\begin{gather}
    \mathbb E_{x_{n,\mu}, z_{n,\mu}} |\gamma^\mu_n  -\mathbb E_{x_{n,\mu}, z_{n,\mu}} \gamma^\mu_n|^3 \leq \frac{2\sqrt 2}{\sqrt{\pi}} \big( x_{n,+}^\top \Psi_n x_{n,+}\big)^{3/2} + 
    |\mathbb E_{x_{n,\mu}, z_{n,\mu}} \gamma^\mu_n|^3.
\end{gather}

We now define the following random variables: 

\begin{gather}
    \sigma_n := r_n^\top \Omega_n r_n, 
    \quad 
    \kappa_n := x_{n,+}^\top \Phi_n r_n, 
    \quad 
    \zeta_n := x_{n,+}^\top \Psi_n x_{n,+},
    \quad 
    \eta_n := z_n^\top r_n.
\end{gather}

If we now let $F_n(s)$ denote the cumulative distribution fnction of the standardized sum 

\begin{gather}
    s=\frac {\sum_\mu (\gamma^\mu_n  - \mathbb E_{x_{n,\mu}} \gamma^\mu_n ) } {\sqrt{p(\zeta_n
    - \frac{2 \kappa_n^2}{\pi\sigma_n })}}
\end{gather}
and denote by $N(s)$ the standard normal c.d.f., the Berry-Esseen theorem \cite{raic2019berry} gives 
 
\begin{gather}
    \bigg|F_n(s) - N(s)\bigg| \leq \frac{K}{\sqrt{p}\big(1 - \frac{2\kappa_n^2}{\pi\sigma_n  \zeta_n}\big)^{3/2}}
    \big( 1 +(x_{n,+}^\top \Psi_n x_{n,+})^{-3/2} |\mathbb E_{x_{n,\mu}, z_{n,\mu}} \gamma^\mu_n|^3\big)
    ,
    \label{eq:bound}
\end{gather}
where $K$ is a constant. We can see that the rightmost term contributes at most a constant factor from 

\begin{gather}
    \frac{|\mathbb E_{x_{n,\mu}, z_{n,\mu}} \gamma^\mu_n|^3}{(x_{n,+}^\top \Psi_n x_{n,+})^{3/2}}
    \leq  \frac{M|x_{n,+}^\top \Phi_n r_n|^3}{(r_n^\top \Omega_n r_n)^{3/2}(x_{n,+}^\top \Psi_n x_{n,+})^{3/2}} \leq M' \frac{\phi_{n,\max}^3}{\omega_{n,\min}^{3/2} \psi_{n,\min}^{3/2}} = O(1), 
\end{gather}

for some positive constants $M,M'$. Next we show that the entire right hand side of Eq. \eqref{eq:bound} is $O(n^{-1/2})$ using the following Lemma.

\begin{lemma}
    Let $C$ by a positive definite matrix with submatrices
    \begin{gather}
        C = 
        \begin{pmatrix} 
        \Psi & \Phi 
        \\ \Phi^\top & \Omega 
        \end{pmatrix},
\end{gather}
with $\Psi \in \mathbb R^{n \times n}, \Omega \in \mathbb R^{d\times d},$ and $\Phi \in \mathbb{R}^{n\times d}$. Then for all $h \in \mathbb R^n$ and $r \in \mathbb R^d$ 
\begin{gather}
    \frac{(h^\top \Phi r)^2}{h^\top \Psi h r^\top \Omega r }\leq 1. 
\end{gather}
\label{lem:ineq}
\end{lemma}
\begin{proof}
    Let $(x,z)$ be mean-zero jointly Gaussian vectors with covariance matrix $C$. It follows that
    \begin{align}
        (h^\top \Phi r)^2 &= \big( \mathbb E[ \langle x, h \rangle \langle r, z \rangle])^2 
        \\
        &\leq \mathbb E[\langle x,h \rangle^2] \mathbb E[\langle z, r \rangle^2]
        \\
        &= (h^\top \Psi h )(r^\top \Omega r).
    \end{align}
\end{proof}

We therefore obtain

\begin{gather}
    \frac{2\kappa_n^2}{\pi\sigma_n  \zeta_n} \leq \frac 2 \pi ,
\end{gather}

so that the right hand side of Eq. \eqref{eq:bound} is $O(n^{-1/2})$. Keeping only the leading order terms, we obtain the following expression for the generalization error

\begin{gather}
    \frac 1 2 \mathbb E_{r_n} \mathbb E_{x_{n}, z_{n}} 
    \erfc\bigg(\frac{\sgn(\eta_n)\kappa_n}{\sqrt{\frac \pi p\sigma_n  \zeta_n(1-\epsilon)}}\bigg) 
    + O (n^{-1/2}) ,
    \\
    \epsilon := \frac{2 \kappa_n^2}{\pi\sigma_n  \zeta_n}.
\end{gather}
Note that we have dropped the $+$ subscript from the test point and latent pair. We can deal with the term $\epsilon \leq 2/\pi$ by noting that 

\begin{gather}
    \mathbb E_{x_n, r_n} \frac{2 \kappa_n^2}{\pi\sigma_n  \zeta_n} \leq 
   \mathbb E_{s,r_n} \frac{2 (s^\top \Psi^{1/2} \Phi_n r_n)^2}{\pi d \psi_{n,\min} \omega_{n,\min}}
    = O(n^{-1}), 
\end{gather}

where $\psi_{n,\min}, \omega_{n,\min}$ are the smallest eigenvalues of $\Omega_n, \Psi_n$, and $s$ is a random vector drawn uniformly from the surface of the unit sphere. From here we can simply expand to first order to see that the $\epsilon$ term contributes at most $O(n^{-1}).$ 

\indt We now show that we may replace the quadratic form involving $x_n$ with its expected value. To do this, we first note that we can focus on estimating the expectation over a region which in which $|\zeta_n - \tr(\Psi_n^2)| > c\tr(\Psi_n^2)$ for some $0<c<1$ by applying the Hanson-Wright inequality. This gives 

\begin{gather}
    \bigg |\mathbb E_{r_n, x_n, z_n} \bigg[\erfc\bigg(
    \frac{\sgn(\eta_n )\kappa_n}{\sqrt{\frac \pi p \sigma_n\tr(\Psi_n^2)}}
    \bigg) - \erfc\bigg(
    \frac{\sgn(\eta_n )\kappa_n}{\sqrt{\frac \pi p \sigma_n\zeta_n}}
    \bigg)\bigg ] \bigg| 
    \\
    \leq  
   \mathbb E_{r_n, x_n, z_n}
   \mathbf{1}_{\{|\zeta_n - \tr(\Psi_n^2) | < c\tr(\Psi_n^2)\}}
    \bigg| \erfc\bigg(
    \frac{\sgn(\eta_n )\kappa_n}{\sqrt{\frac \pi p \sigma_n\tr(\Psi_n^2)}}
    \bigg) - \erfc\bigg(
    \frac{\sgn(\eta_n )\kappa_n}{\sqrt{\frac \pi p \sigma_n\zeta_n}}
    \bigg) \bigg|
     + O(e^{-g \sqrt n}) ,
    \nonumber
\end{gather}
for some constant $g$ that depends on the choice of $c$. We can then obtain the following bound for the integrand in the region where $|\zeta_n - \tr(\Psi_n^2)| < c\tr(\Psi_n^2)$ by the mean value theorem:

\begin{gather}
     \mathbb E_{r_n, x_n, z_n}
   \mathbf{1}_{\{|\zeta_n - \tr(\Psi_n^2) | < c\tr(\Psi_n^2)\}}\bigg | \erfc\bigg(
    \frac{\sgn(\eta_n )\kappa_n}{\sqrt{\frac \pi p \sigma_n\tr(\Psi_n^2)}}
    \bigg) - \erfc\bigg(
    \frac{\sgn(\eta_n )\kappa_n}{\sqrt{\frac \pi p \sigma_n\zeta_n}}
    \bigg)\bigg |
    \\ \leq 
    \mathbb E \frac{\sqrt{p}C|\kappa_n |}{\tr(\Psi_n^2)^{3/2}} \big | \zeta_n - \tr(\Psi_n^2)\big| 
    \leq \frac{\sqrt p C \sqrt{\mathbb E[\kappa_n^2] \mathbb E |\zeta_n - \tr(\Psi_n^2)|^2}}{\tr(\Psi_n^2)^{3/2}}. 
\end{gather}

for some constant $C$. We can now use the following Lemma to show that the right hand side is $O(n^{-1/2})$ in expectation: 

\begin{lemma}
    Let $r \in \mathbb{R}^d$ be a vector distributed  either uniformly on the surface of the sphere of radius $\sqrt d$ or according to a mean-zero Gaussian with covariance $I_d$. Then for any matrix $A \in \mathbb{R}^{d\times d}$ we have 

    \begin{gather}
        \mathbb E|r^\top A r - \tr(A)| = O(||A||_{\mathrm{op}}) + O(||A||_F)  ,
        \\
        \mathbb E|r^\top A r - \tr(A)|^2 = O(||A||^2_{\mathrm{op}}) + O(||A||^2_F)
    \end{gather}

    where $|| \cdot ||_{\mathrm{op}}$ is the operator norm induced by the Euclidean norm. 
    \label{lem:qform}
\end{lemma}

\begin{proof}
    Applying the Hanson-Wright inequality for random vectors satisfying the convex concentration property \cite{adamczak2015note, vershynin2018high} we obtain the first identity

    \begin{gather}
        \int p(dr) |r^\top A r - \tr(A)| 
        = 
        \int_0^{\infty}
        dt
        \mathbb P(|r^\top A r - \tr(A) | > t)
        \\\leq \int_0^\infty dt \exp\bigg(-\frac{C t^2}{||A||_F^2}\bigg) + \exp\bigg( - \frac{c t}{||A||_{\mathrm{op}}}\bigg)  = O(||A||_F) + O(||A||_{\mathrm{op}}),
    \end{gather}
where $C, c$ are positive constants. Similarly we have 
    \begin{gather}
        \int p(dr) |r^\top A r - \tr(A)|^2
        = 
        2\int_0^{\infty}
        dt
        \mathbb P(|r^\top A r - \tr(A) | > t) t
          = O(||A||_F^2) + O(||A||^2_{\mathrm{op}}),
    \end{gather}
yielding the second identity.
\end{proof}

Applying this estimate and using the fact that $\mathbb E \kappa_n^2 = \tr(\Phi_n^\top \Psi_n \Phi_n) = O(n^{-2})$ then shows that the right hand side is $O(n^{-1/2})$. We are therefore left with

\begin{gather}
    \frac 1 2 \mathbb E_{r_n} \mathbb E_{\eta_n, \kappa_n} \erfc\bigg(
    \frac{\sgn(\eta_n) \kappa_n}{\sqrt{\frac \pi p \sigma_n\tr(\Psi_n^2)}}
    \bigg) + O(n^{-1/2}). 
    \label{eq:bivariate}
\end{gather}

From our definitions, the variables $\eta_n, \kappa_n$ follow a bivariate Gaussian distribution: 

\begin{gather}
    \eta_n, \kappa_n \sim \mathcal{N} \bigg(0, 
    \begin{pmatrix}
        r_n^\top \Omega_n r & r_n^\top \Phi^\top_n \Phi_n r_n 
        \\
        r_n^\top \Phi^\top_n \Phi_n r_n & r_n^\top \Phi_n^\top \Psi_n \Phi_n r_n 
    \end{pmatrix}
    \bigg) .
\end{gather}

The inner expectation can now be carried out analytically. We start by integrating over the conditional distribution $\kappa_n | \eta_n$, which is Gaussian with mean $\eta_n r_n^\top \Phi_n^\top \Phi_n r /\sigma_n $ and variance $r_n^\top \Phi_n^\top \Psi_n \Phi_n r_n  - (r_n^\top \Phi^\top_n \Phi_n r_n)^2/\sigma_n$. To do this, we need the Gaussian integral over the complementary error function
\begin{gather}
    \frac 1 2\int \frac{dx}{\sqrt{2\pi v}}e^{-(x-m)^2/(2v)}\mathrm{erfc}(cx)
=
\frac 1 2 \mathrm{erfc}\bigg(\frac{mc}{\sqrt{1+2vc^2}}\bigg).
\label{eq:gauss-erfc}
\end{gather}

This formula can be derived by considering the function 

\begin{gather}
    I(a,b) := \int Dx \erfc(ax+b) ,
\end{gather}
where $Dx$ is a standard Gaussian measure. This integral can be solved by differentiating under the integral with respect to $b$, leading to Eq. \eqref{eq:gauss-erfc}. Applying this identity to Eq. \eqref{eq:bivariate} we obtain 

\begin{gather}
    \frac 1 2 \mathbb E_{r_n} \mathbb E_{\eta_n} \erfc \bigg( 
    \frac{|\eta_n| r_n^\top \Phi_n^\top \Phi_n r_n}{r_n^\top \Omega_n r_n 
    \sqrt{\frac \pi p\sigma_n  \tr(\Psi_n^2) 
    + 2 \big[r_n^\top \Phi_n^\top \Psi_n \Phi_n r_n 
    - \frac{(r_n^\top \Phi_n^\top \Phi_n r_n)^2}{r_n^\top \Omega_n r_n}\big]
    }
    }
    \bigg) .
\end{gather}

We can now split the integral over the $\eta_n$ into half Gaussian integrals over complementary error functions. This allows us to invoke the relation
\begin{gather}
    \frac 1 2 \int_0^\infty \frac{\sqrt{2} dx}{\sqrt{\pi}}e^{-x^2/2} \mathrm{erfc}(cx)
= \frac 1 \pi \tan^{-1} \bigg(\frac 1 { \sqrt{2} c} \bigg).
\end{gather}
This formula can similarly be derived by considering the function 
\begin{gather}
    J(c) := \int_0^\infty Dx \erfc(cx) 
\end{gather}
and differentiating with respect to $c$. This leaves

\begin{gather} 
    \frac 1 \pi \mathbb E_{r_n} \tan^{-1}\bigg( 
    \sqrt{\frac{{r_n^\top \Omega_n r_n\big[
    \frac \pi p \tr(\Psi_n^2)r_n^\top \Omega_n r_n  
    + 2 r_n^\top \Phi_n^\top \Psi_n \Phi_n r_n 
    \big]}}{2(r_n^\top \Phi_n^\top \Phi_n r_n)^2}
    - 1 }
    \bigg) 
    + O(n^{-1/2}) .
    \end{gather} 
Note that \emph{this gives the generalization error for a fixed task vector $r_n$}. 

\indt To carry out the final expectation, we expand these quadratic forms around their mean values. To do this, let us introduce the $O(1)$ variables

\begin{gather}
    \sigma^n_1 = r_n^\top \Omega_n r_n ,
    \quad 
    \sigma^n_2 = n^2 r_n^\top \Phi_n^\top \Psi_n \Phi_n r_n,
    \quad 
    \sigma^n_3 = n r_n^\top \Phi_n^\top \Phi_n r_n ,
\end{gather}
and the function 

\begin{gather}
    G(\sigma^n) = \tan^{-1}\bigg(\sqrt{\frac{\sigma^n_1 (c \sigma^n_1 + 2n^{-2}\sigma^n_2)}{2n^{-2}(\sigma^n_3)^2} -1}\bigg) ,
\end{gather}

where $c:=\frac \pi p \tr(\Psi_n)^2=O(n^{-2})$. We start by noting that we can focus on estimating the expectation over a region in which $||\sigma^n - \mathbb E \sigma^n || < t$ for $t>0$. To see this, note that by the Hanson-Wright inequality: 

\begin{gather}
    \mathbb E_{r_n} G( \sigma^n) = \mathbb E_{r_n} \mathbf{1}_{\{ ||\sigma^n - \mathbb E \sigma^n||<t\}} G(\sigma^n) + O(e^{-c\sqrt n}) .
\end{gather}

We can now bound the remaining error incurred by replacing $\sigma^n$ with its expected value as: 
\begin{align}
    |G(\sigma^n) - G(\mathbb E \sigma^n)| 
    &\leq \sum_i \sup_{||\xi - \mathbb E \sigma^n || < t } |\partial_i G(\xi)||\sigma^n_i - \mathbb E \sigma^n_i| 
    \\
    &= \sum_i  \sup_{||\xi - \mathbb E \sigma^n || < t }\bigg |\bigg[\frac{\xi_1(c\xi_1 + 2n^{-2}\xi_2)}{2n^{-2}(\xi_3)^2}-1\bigg]^{-1/2} \bigg[\frac{\xi_1(c\xi_1 + 2n^{-2}\xi_2)}{2n^{-2}(\xi_3)^2} \bigg]^{-1} \\
\\ 
&\partial_i \frac{\xi_1(c\xi_1 + 2n^{-2} \xi_2)}{2n^{-2}(\xi_3)^2}\bigg| |\sigma^n_i - \mathbb E \sigma^n_i| 
    \\
    &\leq M    \sup_{||\xi - \mathbb E \sigma^n || < t }\bigg|\frac{\xi_1(c\xi_1 + 2n^{-2}\xi_2)}{2n^{-2}(\xi_3)^2}-1\bigg|^{-1/2}   \sum_i |\sigma^n_i - \mathbb E \sigma^n_i| ,
\end{align}

for some constant $M$. We can now use the following lemma to give a uniform bound on the remaining term: 
\begin{lemma}

     Let $C$ by a positive definite matrix with submatrices
    \begin{gather}
        C = 
        \begin{pmatrix} 
        \Psi & \Phi 
        \\ \Phi^\top & \Omega 
        \end{pmatrix},
\end{gather}
with $\Psi \in \mathbb R^{n \times n}, \Omega \in \mathbb R^{d\times d},$ and $\Phi \in \mathbb{R}^{n\times d}$. Then 

\begin{gather}
    \frac{\tr(\Omega) \tr(\Phi^\top \Psi \Phi)}{\tr(\Phi^\top \Phi)^2} \geq 1. 
\end{gather}
\end{lemma}

\begin{proof}
    The claim follows from an application of the Cauchy-Schwarz inequality:
    \begin{gather}
        \tr(\Phi^\top \Phi)^2 = (\mathbb E \langle \Phi^\top x, z \rangle)^2  \leq \mathbb E || \Phi^\top x||^2 \mathbb E ||z||^2 = \tr(\Phi^\top \Psi \Phi) \tr(\Omega) , 
\end{gather}
where the expectation is taken over zero-mean, jointly Gaussian vectors $(x,z)$ which have a covariance matrix $C$. 
\end{proof}
From this lemma and our assumptions on the spectrum, it follows that there exists a constant $M'>0$ that does not depend on $n$ such that

\begin{gather}
    \frac{\langle \sigma^n_1 \rangle (c \langle \sigma^n_1 \rangle + 2n^{-2}  \langle \sigma^n_2 \rangle)}{2n^{-2} \langle \sigma^n_1 \rangle^2} \geq 1 + M',
\end{gather}
where we use brackets to denote the expectation. Hence if we choose $t$ sufficiently small, we obtain

\begin{gather}
    |G(\sigma^n) - G(\mathbb E \sigma^n)| \leq M'' \sum_i |\sigma^n_i - \mathbb E \sigma^n_i | ,
\end{gather}
where $M''$ is a positive constant. By Lemma \ref{lem:qform}, the right hand side is $O(n^{-1/2})$.

\indt Replacing these quadratic forms with their means, we obtain the final result:  

\begin{gather}
    \frac 1 \pi \tan^{-1}\bigg( \sqrt{
    \frac{\tr(\Omega_n) \big[
    \frac \pi p \tr(\Psi_n^2)\tr(\Omega_n) 
    + 2 \tr(\Phi_n^\top \Psi_n \Phi_n) 
    ]}{2\tr(\Phi_n^\top \Phi_n)^2}
    -1}
    \bigg) 
    + O(n^{-1/2}).
    \label{eq:final-eg}
\end{gather}
\end{proof}

\section{Decomposition of generalization error into geometric terms}

Having obtained a closed-form solution for generalization error in the thermodynamic limit, we now seek to rewrite this formula as a monotone function of interpretable geometric terms.
We start by splitting Eq. \eqref{eq:final-eg} into a $p$-dependent and $p$-independent term:
\begin{gather}
\frac{1}{\pi}\tan^{-1}\left(\sqrt{
    \frac{\pi}{2p}\frac{ 
\mathrm{Tr}(\Omega_n) 
\mathrm{Tr}(\Psi_n^2) \mathrm{Tr}(\Omega_n)}{(\Tr(\Phi_n\Phi_n^\top))^2}
+ \frac{\Tr(\Omega_n)\mathrm{Tr}(\Phi_n^\top\Psi_n\Phi_n)}
{(\tr(\Phi_n\Phi_n^\top))^2}
-1}\right).
\end{gather}
If we re-write the $p$-dependent term as
\begin{gather}
\frac{\pi}{2p} \left(\frac{\Tr(\Omega_n)\Tr(\Psi_n)}{\Tr(\Phi_n\Phi_n^\top)}\right)^2 \cdot \frac{\Tr(\Psi_n^2)}{\Tr(\Psi_n)^2},
\end{gather}
we can notice the \textbf{participation ratio of the neural activity} 
\begin{gather}
\mathrm{PR}(\Psi) = \frac{\Tr(\Psi_n)^2}{\Tr(\Psi_n^2)},
\end{gather} 
which is a well-known measure of the dimensionality of neural activity. We collect the rest of the $p$-dependent term into our \textbf{correlation term}
\begin{gather}
c = \frac{\Tr(\Phi_n\Phi_n^\top)}{\Tr(\Omega_n)\Tr(\Psi_n)},
\end{gather}
and the generalization error now reduces to
\begin{gather}
\frac{1}{\pi}\tan^{-1}\left(\sqrt{
    \frac{\pi}{2pc^2\PR(\Psi)}
+ \frac{\Tr(\Omega_n)\mathrm{Tr}(\Phi_n^\top\Psi_n\Phi_n)}
{(\tr(\Phi_n\Phi_n^\top))^2}
-1}\right).
\end{gather} 
To interpret the correlation term, note that the numerator $\Tr(\Phi_n\Phi_n^\top)$ can be expanded as
\begin{gather}
\Tr(\Phi_n\Phi_n^\top) = \sum_{i=1}^d\sum_{j=1}^n \mathbb E[z_ix_j]^2,
\end{gather}
which is a sum-of-squares of the covariance between all pairs of latent variables $z_i$ and all components $x_j$ of neural activity. The denominator of $c$ is just the total latent variance $\Tr(\Omega_n)$ times the total variance $\Tr(\Psi_n)$ of neural activity. The correlation term $c$ can therefore be viewed as a generalization of the well-known Pearson correlation between two scalar variables to capture the total correlation strength between the $d$-dimensional latent variable $z$ and the $n$-dimensional neural activity $x$. In particular, when $n,d = 1$, the total correlation $c$ reduces to the square of the standard Pearson correlation between $z$ and $x$.

\indt We call (the inverse of) the remaining $p$-independent term the \textbf{alignment term} $a$:
$$a = \frac{(\Tr(\Phi_n\Phi_n^\top))^2}{\Tr(\Omega_n)\Tr(\Phi_n^\top\Psi_n\Phi_n)}.$$
In the main text (and as described below), the alignment term is subdivided further, but we can also interpret $a$ as a whole. We can expand the main term in the denominator as
\begin{gather} 
\Tr(\Phi_n^\top \Psi_n \Phi_n) = \sum_{i=1}^d \phi_i^\top \Psi_n \phi_i = \sum_{i=1}^d \|\phi_i\|^2 \operatorname{Var}(\hat \phi_i \cdot x),
\end{gather} 
where $\phi_i$ is the $i^\mathrm{th}$ column of $\Phi_n$, and $\operatorname{Var}(\hat \phi_i \cdot x)$ is the variance of the total neural response $x$ projected onto the direction of $\phi_i$. Using the fact that the average neural activity, conditioned on a specific value of the latents is given by $\mathbb E[x|z] = \Phi \Omega^{-1} z$, we can see that for a diagonal $\Omega$, the column $\hat \phi_i$ corresponds to the \emph{coding direction of the latent variable $z_i$}. Additionally, the norm $\|\phi_i\|$ reflects the \emph{coding strength of the latent variable $z_i$}.  The numerator of $a$ can be written as
\begin{gather}
(\Tr(\Phi_n\Phi_n^\top))^2 = \left(\sum_{i=1}^d \|\phi_i\|^2\right),
\end{gather}
and therefore depends only on the norms of the coding directions $\phi_i$ but not on their angles. Holding the norms of each $\phi_i$ fixed, we can see that the alignment term $a$ prefers that each coding direction $\hat \phi_i$ be positioned along a direction of minimal response variance in the neural state space. In other words, the only variance along a direction $\hat \phi_i$ should be caused by variations in the corresponding latent variable.
Intuitively, the alignment term $a$ encourages arranging the coding directions of independent latent variables in such a way that the signal for these variables do not interfere with one another and so that all latents are coded along directions with low intrinsic noise. We make this notion more formal below.

\indt Now let us subdivide the alignment term $a$ as in the main text. We partition the total neural variance $\Psi$ into stimulus-driven variance and stimulus-independent variance. Recall that in our setup, latents $z$ and neural responses $x$ are drawn from a joint Gaussian distribution
\begin{gather}
(x, z) \sim \mathcal{N}\left(0, \begin{pmatrix}\Omega & \Phi^\top \\ \Phi & \Psi\end{pmatrix}\right).
\end{gather} 
This is equivalent to a model in which we first sample $z \sim \mathcal{N}(0, \Omega)$ and then sample $x$ as $x = \Phi\Omega^{-1} z + \varepsilon$, where $\varepsilon$ is stimulus-independent Gaussian noise with a covariance given by
\begin{gather}
    \Psi - \Phi \Omega^{-1} \Phi^\top ,
\end{gather}
which is the covariance of the neural responses conditional on the latents, $\mathrm{cov}(x|z)$. This suggests a partitioning of the total neural variance $\Psi$ into stimulus-independent variance, $H := \Psi - \Phi\Omega^{-1}\Phi^\top$, and stimulus-driven variance, $\mathrm{cov}(x) - H = \Phi \Omega^{-1} \Phi^\top$.

\indt We can now decompose the $p$-independent term of our generalization error formula as
\begin{gather}
\frac{\Tr(\Omega_n)\Tr(\Phi_n^\top \Psi_n \Phi_n)}{(\Tr(\Phi_n\Phi_n^\top))^2} = \frac{\Tr(\Omega_n)\Tr(\Phi_n^\top H \Phi_n)}{(\Tr(\Phi_n\Phi_n^\top))^2} + \frac{\Tr(\Omega_n)\Tr(\Phi_n^\top (\Phi\Omega^{-1}\Phi^\top) \Phi_n)}{(\Tr(\Phi_n\Phi_n^\top))^2}.
\end{gather} 
We call the inverse of the first term \textbf{signal-noise factorization} $s$:
\begin{gather}
\frac{1}{s} = \frac{\Tr(\Omega_n)\Tr(\Phi_n^\top H\Phi_n)}{(\Tr(\Phi_n\Phi_n^\top))^2}.
\end{gather}
This interpretation can be understood by noting that
\begin{gather}
\Tr(\Phi_n^\top H \Phi_n) = \sum_{i=1}^d \phi_i^\top H \phi_i 
= \sum_{i=1}^d \|\phi_i\|^2 \cdot \hat\phi_i^\top H \hat\phi_i.
\end{gather}
We can see that $\hat\phi_i^\top H \hat\phi_i$ gives the projection of stimulus-independent noise corrupting the signal for the $i^\mathrm{th}$ latent variable, so $\Tr(\Phi_n^\top H \Phi_n)$ measures the amount of stimulus-independent noise corrupting the signal directions for each latent variable. In particular, note that $s$ is maximized when there is no stimulus-independent noise in the neural dimensions used for coding the latent variables.


\indt Finally, we call the inverse of the second part of the $p$-independent term the \textbf{signal-signal factorization}:
\begin{gather}
    \frac{1}{f} = \frac{\Tr(\Omega_n)\Tr(\Phi_n^\top (\Phi\Omega^{-1}\Phi^\top) \Phi_n)}{(\Tr(\Phi_n\Phi_n^\top))^2}.
\end{gather}
To understand this term, assume without loss of generality that 
\begin{gather}
\Omega = \operatorname{diag}(\omega_1, \dots, \omega_D),
\end{gather}
which can be accomplished by making an orthogonal change of variables. Now observe that
\begin{gather}
    \Tr(\Phi_n^\top\Phi_n\Omega^{-1}\Phi_n^\top\Phi_n) = \sum_{i=1}^d\sum_{j=1}^d \frac{1}{\omega_i} \langle \phi_i, \phi_j \rangle^2.
\end{gather}
The numerator of $1/f$ is a weighted sum of the dot products between \emph{all pairs} of coding vectors $\phi_i$ and $\phi_j$ for latents $z_i$ and $z_j$. By contrast, the denominator $(\Tr(\Phi_n\Phi_n^\top))^2$ can be expanded to
\begin{gather}
    (\Tr(\Phi_n\Phi_n^\top))^2 = \left(\sum_{i=1}^d\langle \phi_i, \phi_i \rangle\right)^2,
\end{gather}
which only captures the signal strengths $\|\phi_i\|$ without regard for the angles between pairs of distinct signal directions. With the norms $\|\phi_i\|$ fixed, $f$ is maximized by making all coding directions orthogonal. That is, $f$ is maximized by a factorized code.

\indt The interpretation of our factorization term $f$ is somewhat complicated by the fact that it depends on the relative norms $\|\phi_i\|$ of the columns of $\Phi$ in addition to the coding directions $\hat \phi_i$. Some such dependence is inevitable. For example, if one coding direction $\phi_i$ has norm very near zero (i.e., has near zero signal strength to begin with), then the angle it makes with other coding directions is irrelevant to generalization error and should not enter into $f$. However, we will demonstrate numerically that practically all of the variation we observe in $f$ between the layers of artificial neural networks (Fig. 5 and SM Fig. \ref{fig:mlp-tanh}), during training of artificial neural networks (Fig. 6 and SM Fig. \ref{fig:mlp-dynamics-tanh}), and between brain regions (Fig. 7 and SM Fig. \ref{fig:majaj-global}) is driven by changes in the \emph{angles} between coding directions. To this end, we obtain a simplified analog of $1/f$ by assuming that all signal directions have the same norm and assuming that $\Omega$ has a flat spectrum. This yields the simplified quantity
\begin{equation}
\frac{\sum_{i=1}^d \sum_{j=1}^d \langle \hat \phi_i, \hat \phi_j \rangle}{d},
\end{equation}
which manifestly only depends on the angles between coding directions. In fact, this simplified version of $1/f$ is linearly related to the mean-squared cosine similarity between all pairs of distinct coding directions:
\begin{equation}\label{eq:mscs}
\frac{\sum_{i \neq j} \langle \hat \phi_i, \hat \phi_j \rangle}{d(d-1)}.
\end{equation}
We therefore compare the quantity $1/f$ with the mean-squared cosine similarity between all pairs of distinct coding directions (Eq. \ref{eq:mscs}) and show that the two are almost exactly linearly related for all representations we analyzed in both artificial neural networks and neural data (SM Fig. \ref{fig:sm-fact}). 

\begin{figure}
    \centering
    \includegraphics[width=0.95\textwidth]{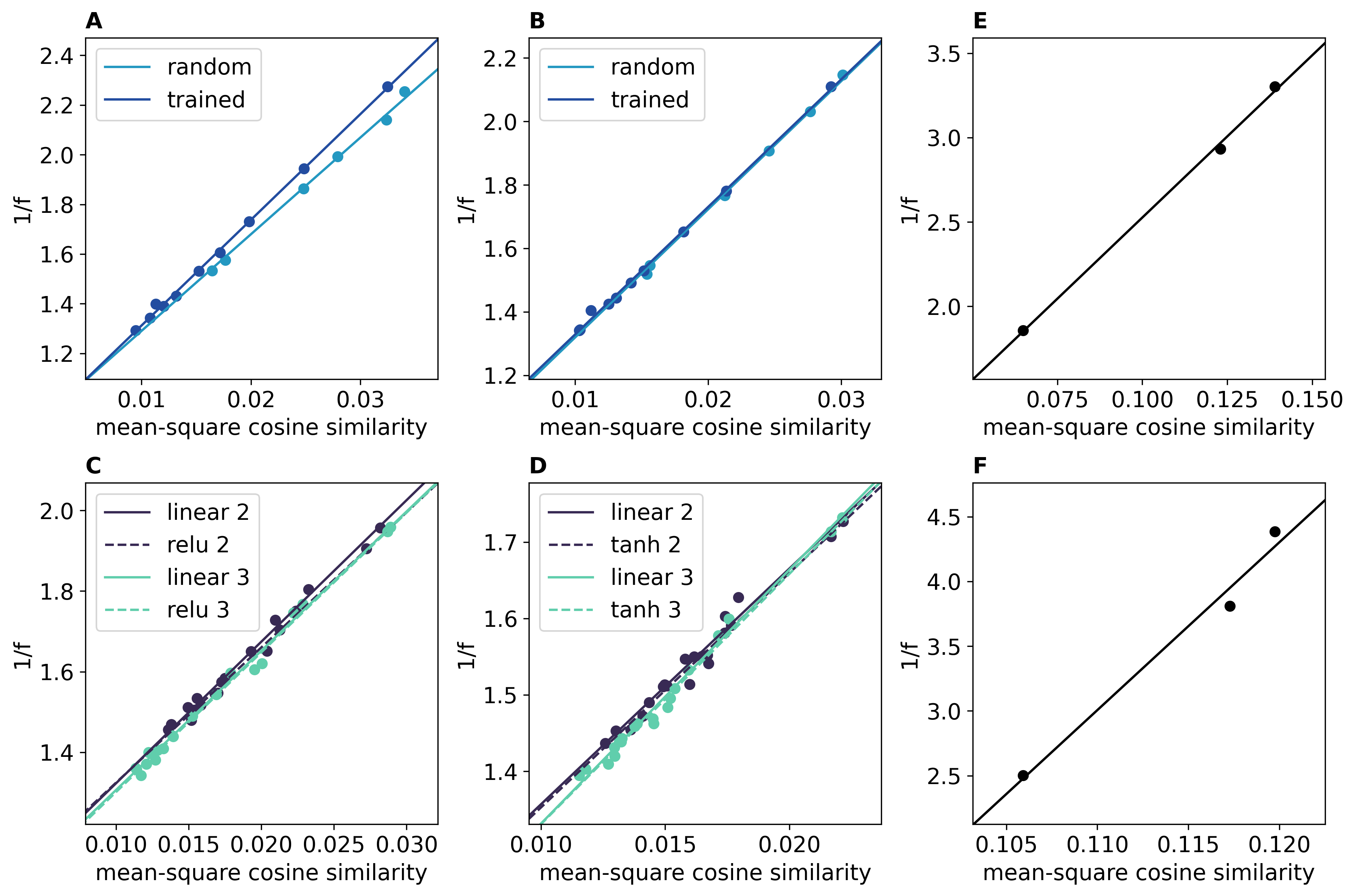}
    \caption{Almost all variation in our factorization term is driven by changing angles between coding directions. We compare the reciprocal $1/f$ of factorization to the mean-square cosine angle between the coding directions for all pairs of distinct latent variables (see Eq. \eqref{eq:mscs}), which only depends on the angles between coding directions. Despite our factorization term $f$ nominally depending on the norms $\|\phi_i\|$ of the coding directions, the almost perfect linear relationship between $1/f$ and mean-square cosine similarity shows that we can effectively consider $f$ to only depend on the angles between coding directions across all the experiments presented in this work: (a) Representations across layers of random and trained ReLU networks (Fig. 5). (b) Representations across layers of random and trained tanh networks (SM Fig. \ref{fig:mlp-tanh}). (c) Representations taken across training epochs in trained ReLU networks (Fig. 6). (d) Representations across training epochs in trained tanh networks (SM Fig. \ref{fig:mlp-dynamics-tanh}). (e) Representations across pixels, V4, and IT in macaque electrophysiological data (measures computed separately for each object category, Fig. 7). (f) Representations across pixels, V4, and IT in macaque electrophysiological data (measures pooled across object categories, SM Fig. \ref{fig:majaj-global}).} 
    \label{fig:sm-fact}
\end{figure}

\section{Optimal geometry}

Here we derive our formula for the spectrum of the optimal neural representation and show that it is disentangled. The argument is not fully rigorous, but it matches numerical results very well. Our goal is to calculate

\begin{gather}
    \arg \min_{\Psi, \Phi} \frac 1 \pi \tan^{-1}\bigg( \frac{ 
\sqrt{\mathrm{Tr}(\Omega) 
\big[
\frac \pi P \mathrm{Tr}(\Psi^2) \mathrm{Tr}(\Omega)
+ 2 \big(\mathrm{Tr}(\Phi^\top\Psi\Phi) 
- \frac{(\mathrm{Tr}[\Phi\Phi^\top])^2}{\mathrm{Tr}(\Omega)}
\big)
\big]
}
}
{\sqrt{2} \mathrm{Tr}(\Phi\Phi^\top)}
\bigg),
\label{eq:optimize}
\end{gather}

subject to 

\begin{gather}
    \begin{pmatrix}
\Omega & \Phi 
\\
\Phi^\top & \Psi 
\end{pmatrix}\succ 0.
\end{gather}

For a positive matrix $\Omega$, this condition is equivalent to

\begin{gather}
    \Psi \succ 0, \quad \Psi - \Phi \Omega^{-1} \Phi^\top \succ 0. 
\end{gather}

Rewriting the argument of the $\tan^{-1}$ function, we can see that the objective may be written as 

\begin{gather}
     \arg \min_{\Psi, \Phi}\frac{\frac{\pi}{P}\Tr(\Psi^2)\Tr(\Omega)+2\Tr(\Phi^\top\Psi\Phi)}{\tr(\Phi\Phi^\top)^2}. 
    \label{eq:simpleobj}
\end{gather}

Our approach is to first minimize Eq. \eqref{eq:simpleobj} holding $\tr(\Phi\Phi^\top)=\gamma$ fixed and show that the objective is ultimately invariant to the choice of $\gamma$. The optimization can be done by introducing the Lagrangian 

\begin{gather}
    L(\Psi, \Phi, \rho, \zeta) =
    c \tr(\Psi^2) + 2 \tr(\Phi^\top \Psi \Phi) 
    + \zeta(\gamma -\tr(\Phi\Phi^\top))
    \\
    - \int d^nx \rho(x) \tr(xx^\top (\Psi - \Phi \Omega^{-1} \Phi^\top)) ,
\end{gather}

where the $\rho(x)$ are KKT multipliers enforcing the positive definite constraint, $\zeta$ is a Lagrange multiplier, and $c:=\frac \pi P\tr(\Omega)$ . Note that we do not explicitly enforce the positivity constraint on $\Psi$, as we find that this is unnecessary. The KKT equations are

\begin{gather}
    \partial_\Psi L = 2c\Psi + 2\Phi \Phi^\top - \langle xx^\top\rangle_\rho=0, 
    \label{eq:psi}
    \\
    \partial_\Phi L = 4 \Psi \Phi  +2 \langle xx^\top\rangle \Phi \Omega^{-1}  + \zeta \Phi  = 0,
    \label{eq:phi}
    \\
    \delta_{\rho} L = \tr(xx^\top (\Psi - \Phi \Omega^{-1} \Phi^\top)) \geq 0 ,
    \\
    \rho(x)\tr(xx^\top (\Psi - \Phi \Omega^{-1} \Phi^\top )) =0 ,
    \\
    \tr(\Phi\Phi^\top) = \gamma.
\end{gather}

We try for a solution with $\Psi = \Phi\Omega^{-1} \Phi^\top.$ Since the variance of $x$ conditional on $z$ is given by $\Psi - \Phi \Omega^{-1} \Phi^\top,$ this is equivalent to looking for solutions in which the neural code has no signal-unrelated noise. Then Eq. \eqref{eq:psi} gives $\langle xx^\top\rangle_\rho = 2 \Phi (c\Omega^{-1} + I) \Phi^\top$. Using these two formulae, Eq. \eqref{eq:phi} gives the condition 

\begin{gather}
    \Phi ( 4 \Omega^{-1} \Phi^\top \Phi + 4(c\Omega^{-1} + I)\Phi^\top \Phi \Omega^{-1} + \zeta I) = 0. 
\end{gather}

This suggests looking for a $\Phi$ whose right singular vectors are aligned with those of $\Omega$. The eigenvalues of $\Phi$ are then given by:

\begin{gather}
    \phi_i^2 = \frac{\gamma}{\sum_i \frac{\omega_i^2}{2\omega_i + \frac{\pi}{P}\tr(\Omega)}}
    \frac{\omega_i^2}{2\omega_i + \frac{\pi}{p}\tr(\Omega)}.
\end{gather}

From the assumption that $\Psi = \Phi^\top \Omega^{-1} \Phi $, we can see that the eigenvectors of $\Psi$ and the left singular vectors of $\Phi$ can be chosen however we want, so long as they are the same. The eigenvalues of $\Psi$ are then simply

\begin{gather}
    \psi_i = \frac{\gamma}{\sum_i \frac{\omega_i^2}{2\omega_i + \frac{\pi}{p}\tr(\Omega)}}
    \frac{\omega_i}{2\omega_i + \frac{\pi}{p}\tr(\Omega)}.
\end{gather}

Plugging this solution into the objective, we indeed find that $\gamma$ drops out of the picture entirely.  Since we are free to choose $\gamma$, we obtain that the optimal representation satisfies the following properties: (1) The eigenvectors of $\Psi$ are the left singular vectors of $\Phi$ and (2) The eigenvectors of $\Omega$ are the right singular vectors of $\Phi$. This implies that the optimal representation is disentangled–i.e., that independent directions in the latent space map onto independent directions in the neural space. 

\indt To see why, consider the covariance between an independent direction in the neural space and an independent direction in the latent space. If we let $o$ be the eigenvector of $\Psi$ corresponding to the neural direction and $u$ an eigenvector of $\Omega$ corresponding to the latent direction, we have

    \begin{gather}
        \mathbb E [\langle o,x\rangle \langle u, z\rangle] = u^\top \Phi o.
    \end{gather}
For the optimal representation, this is either 0 or equal to a singular value of $\Phi.$ Since for Gaussian variables, uncorrelated variables are independent, this argument shows that each independent direction in the neural space is independent from all but one latent dimension. Thus the optimal representation described here is disentangled. 

\indt Finally, the eigen/singular-values of the matrices are given by 

\begin{gather}
    \phi_i^2 \propto \frac{\omega_i^2}{2p\omega_i + \pi\tr(\Omega)},
    \\
    \psi_i \propto \frac{\omega_i}{2p\omega_i + \pi\tr(\Omega)}.
\end{gather}

This gives the optimal representation's spectrum. As stated in the main text, we can see that as $p$ grows relative to $d$, the optimal representation's spectrum becomes increasingly flat, indicating that more variance in the neural state space is being allocated to the less informative directions in the latent space.

\section{Additional MLP analyses}

Here we redo the analyses presented in the main text using a $\tanh$ non-linearity in the random and trained networks.

\begin{figure}
    \centering
    \includegraphics[width=0.95\textwidth]{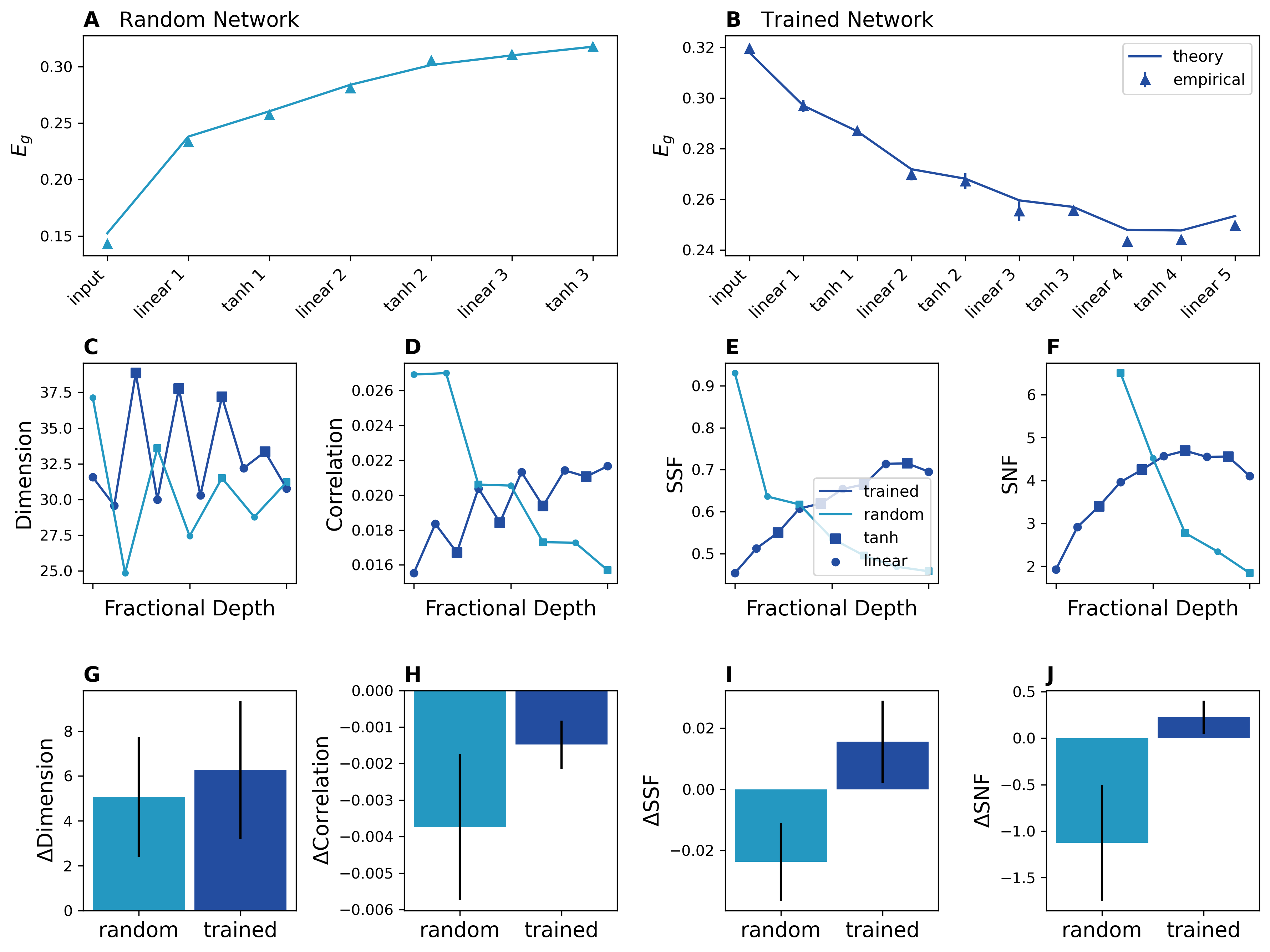}
    \caption{Generalization error and geometry of the $\tanh$ MLPs. (a-b) Generalization error for the random and trained network. (c-f) Layer-wise geometry of the networks. (g-j) Average change in the geometry when passing from a linear to a $\tanh$ layer. We find similar trends as with the relu non-linearity.}
    \label{fig:mlp-tanh}
\end{figure}

\begin{figure}
    \centering
    \includegraphics[width=0.95\textwidth]{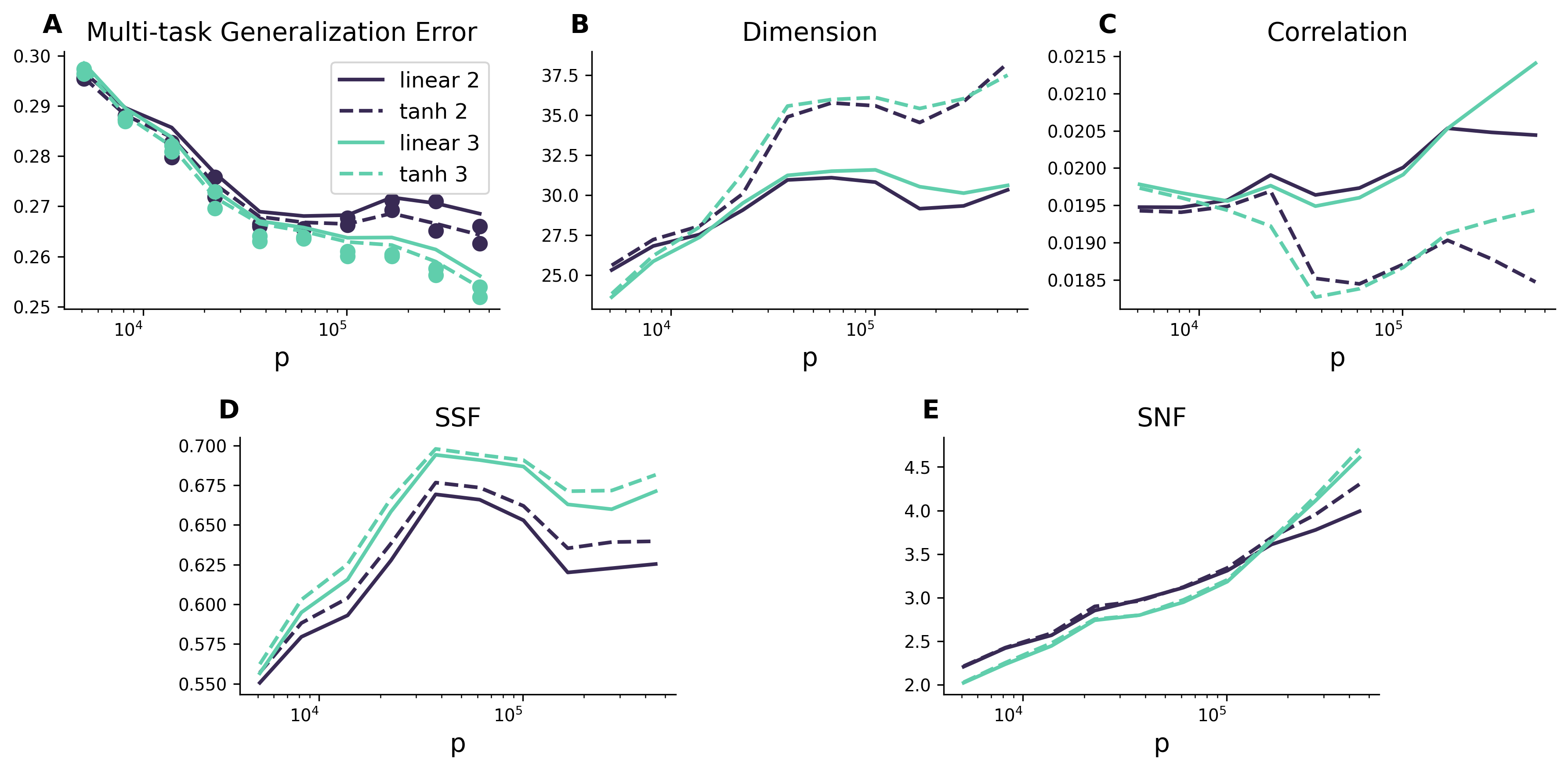}
    \caption{Dynamics of generalization error (a) and geometry (b-e) through training using networks with a $\tanh$ non-linearity. In analogy to the main text, we show 2 early and 2 late layers. At each of the early/late stages, we show a linear and $\tanh$ layer.}
    \label{fig:mlp-dynamics-tanh}
\end{figure}

\clearpage
\newpage 




\clearpage
\newpage

\section{Additional macaque analyses}

\begin{figure}[h]
    \centering
    \includegraphics[width=0.95\textwidth]{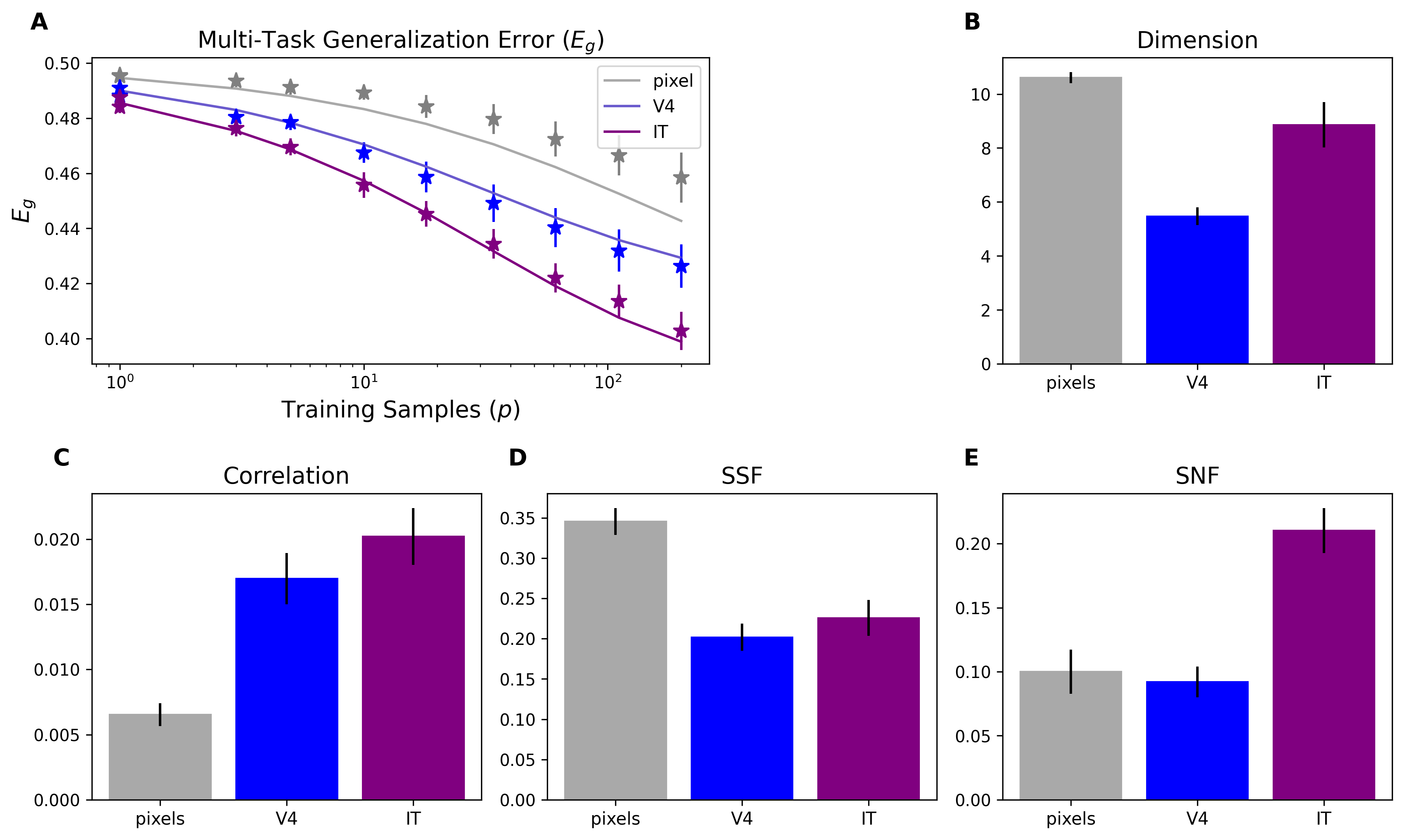}
    \caption{Geometry and generalization error using the same number of units across regions \cite{majaj2015simple}. Here, we repeat the analysis presented in the main text, only we project the pixels and IT data down to 88 dimensions using Gaussian random projection. (a) Generalization error for the three representations. (b-e) Geometric terms. We find qualitatively similar trends to those presented in the main text.}
    \label{fig:majaj-proj}
\end{figure}

\begin{figure}
    \centering
    \includegraphics[width=0.475\textwidth]{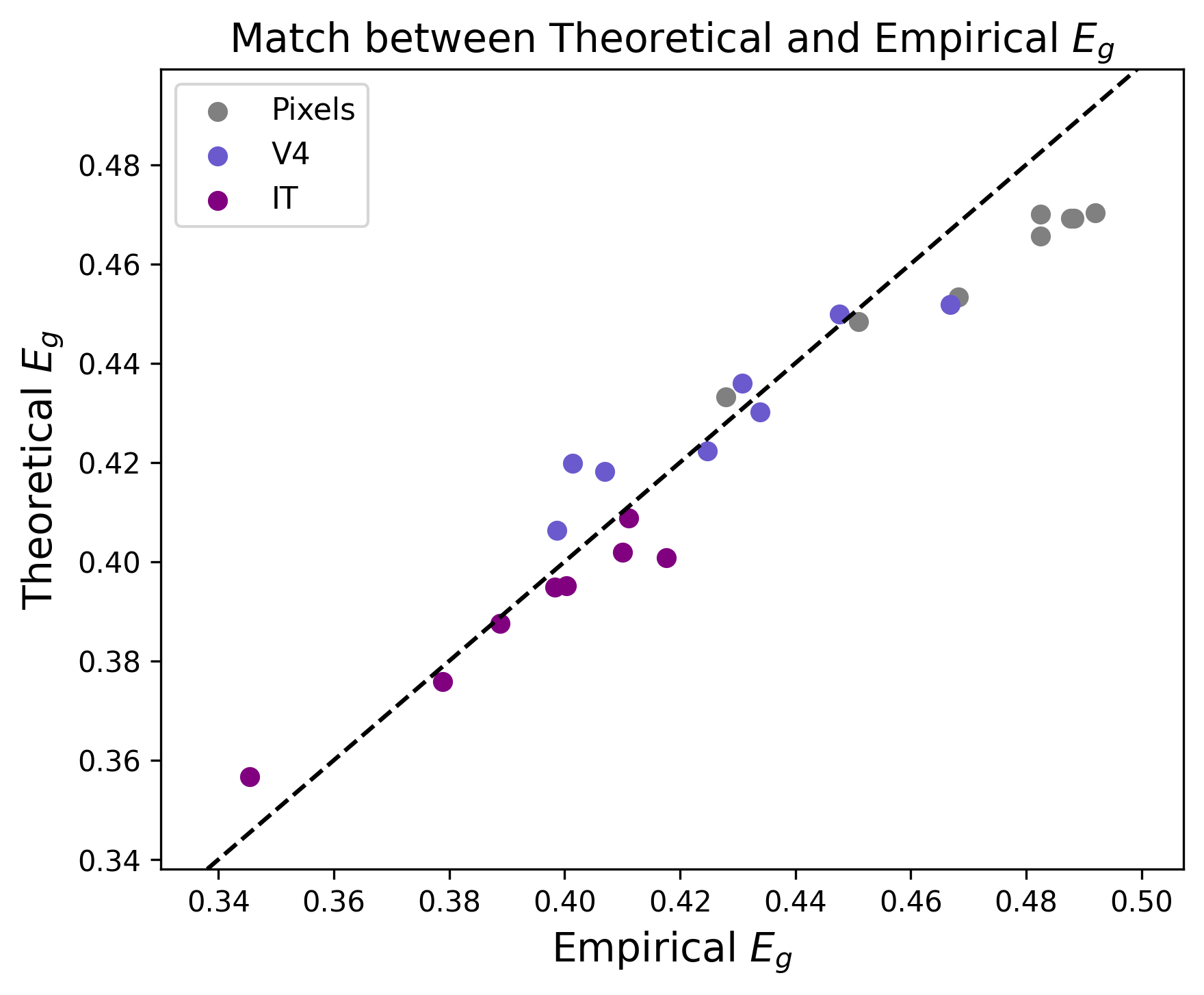}
    \caption{Match between theoretical and empirical generalization error across each of the 8 individual categories. In the main text, we presented the theoretical and empirical generalization errors averaged over individual object categories. Here we show that the theory predicts the empirical generalization error well across individual object categories.}
    \label{fig:majaj-errs}
\end{figure}

\begin{figure}
    \centering
    \includegraphics[width=0.95\textwidth]{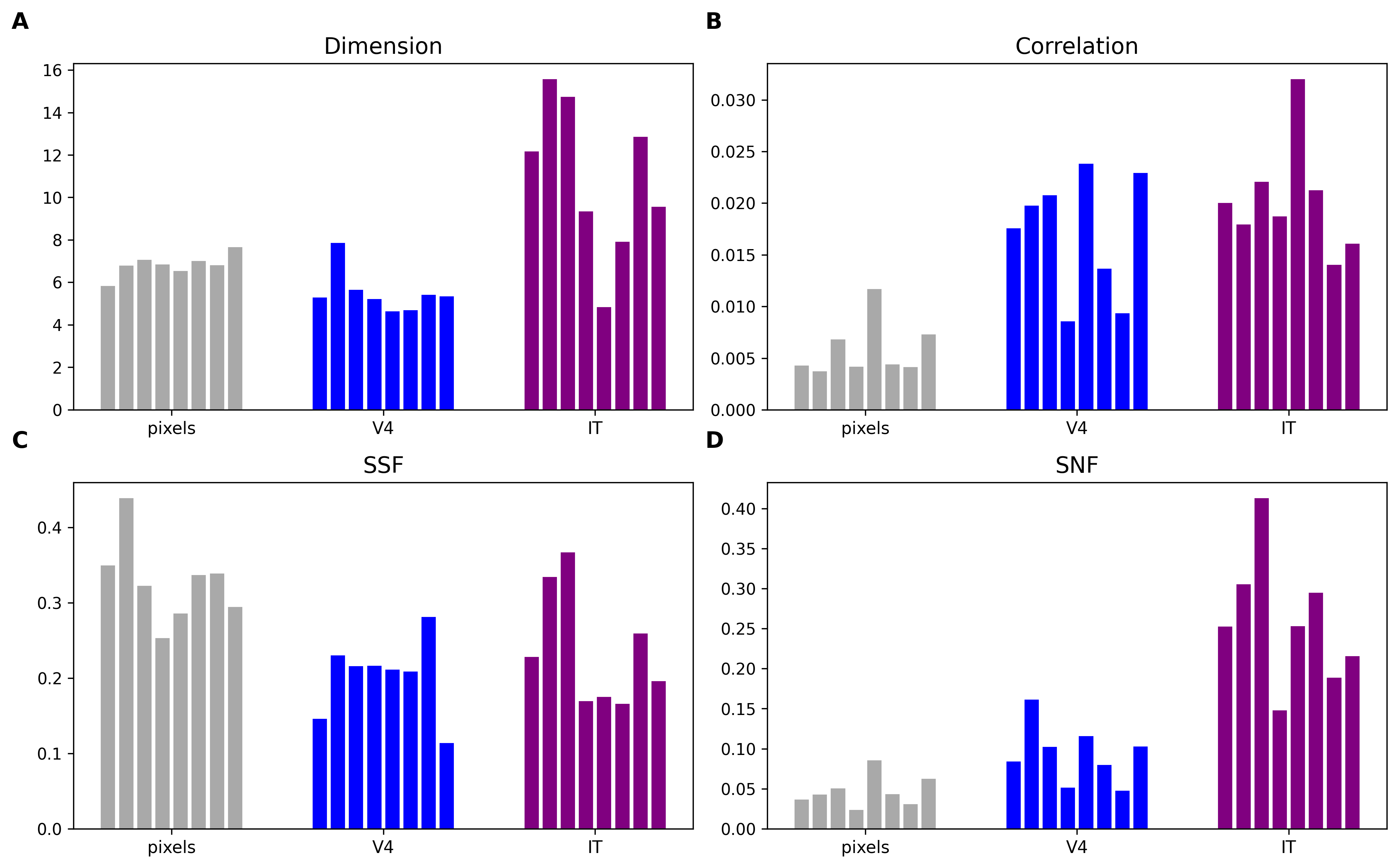}
    \caption{Geometry across individual object categories. Here, we present the distribution of geometric terms, calculated on subsets of the data corresponding to each of the 8 object categories. Note that the terms in the main text correspond to averages over these 8 values.}
    \label{fig:maj-indiv}
\end{figure}

\clearpage
\begin{figure}
    \centering
    \includegraphics[width=0.95\textwidth]{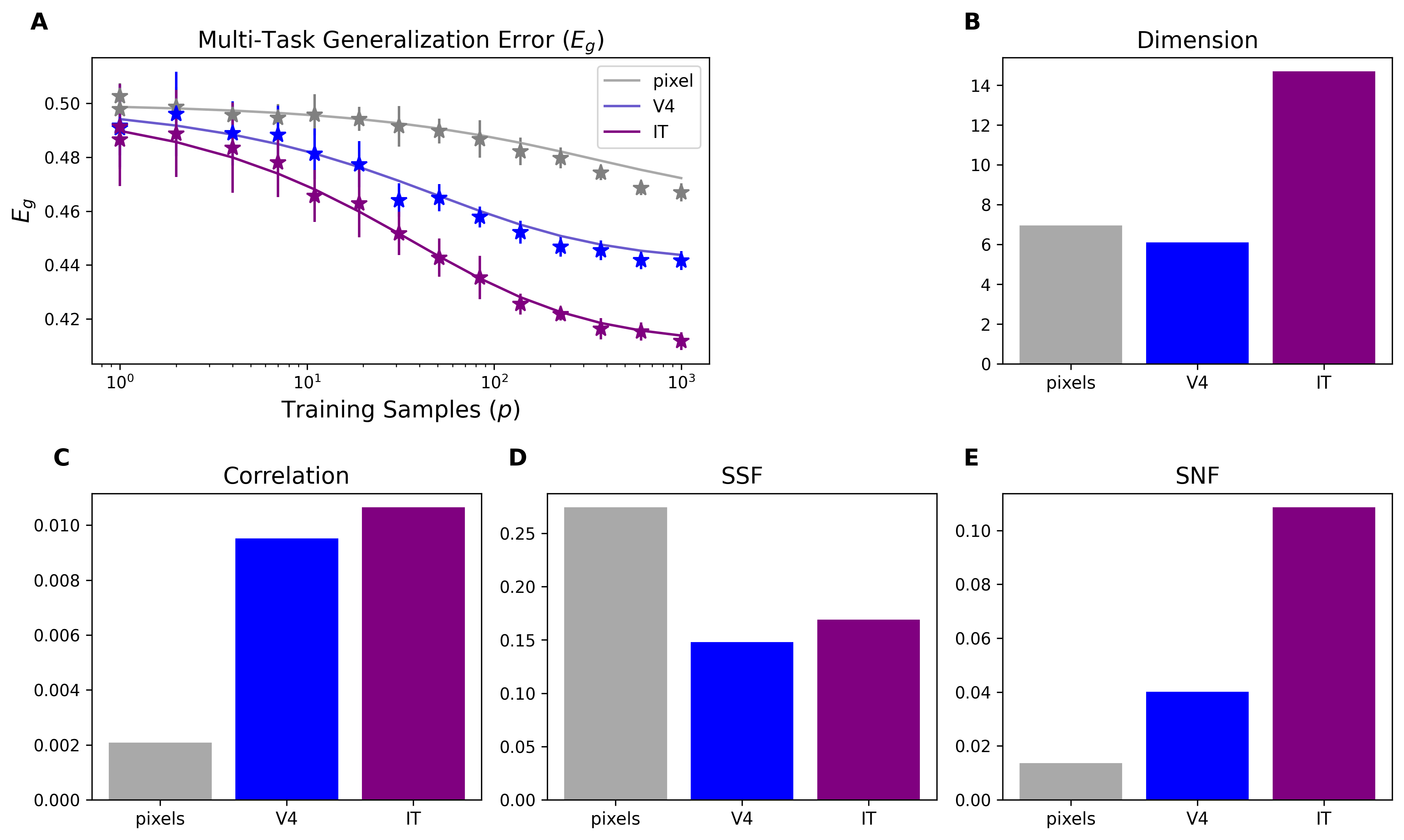}
    \caption{Generalization error and geometry of the pooled monkey data. Here, we pool data from all categories together, rather than considering data subsets corresponding to stimuli coming from the same object category as done in the main text. We can see that the  trends are largely the same and that the theory predicts the empirical error. (a) Theoretical and empirical generalization error. (b-e) Geometry of the pooled data.}
    \label{fig:majaj-global}
\end{figure}

\clearpage
\newpage

\bibliographystyle{unsrt}